\newif\ifincludeproofs

\includeproofstrue 

\ifincludeproofs
\documentclass[letterpaper, 10 pt, conference]{ieeeconf}  
\else
\documentclass[letterpaper, 10 pt, conference]{ieeeconf}  
\fi

\IEEEoverridecommandlockouts                              

\usepackage{amsmath} 
\usepackage{amssymb}  
\usepackage[hidelinks]{hyperref}
\usepackage[capitalize]{cleveref}
\usepackage{cite}

\usepackage{tikz}
\usepackage{siunitx}

\ifincludeproofs
\makeatletter
\newcommand\makebig[2]{%
  \@xp\newcommand\@xp*\csname#1\endcsname{\bBigg@{#2}}%
  \@xp\newcommand\@xp*\csname#1l\endcsname{\@xp\mathopen\csname#1\endcsname}%
  \@xp\newcommand\@xp*\csname#1r\endcsname{\@xp\mathclose\csname#1\endcsname}%
}
\makeatother

\makebig{biggg} {3.0}
\makebig{Biggg} {3.5}
\makebig{bigggg}{4.0}
\makebig{Bigggg}{4.5}
\makebig{biggggg}{5.0}
\fi

\usetikzlibrary{shapes,arrows,fit, calc, angles, quotes, arrows.meta, decorations.markings, decorations.pathreplacing, positioning}

\newtheorem{theorem}{Theorem}

\newtheorem{assumption}{Assumption}
\Crefname{assumption}{assumption}{Assumptions}

\newtheorem{lemma}{Lemma}

\newtheorem{remark}{Remark}

\DeclareMathOperator{\sign}{sign}
\DeclareMathOperator{\diag}{diag}

\title{\LARGE \bf
\ifincludeproofs
\vspace{-7mm}Formation Path Following Control of Underactuated USVs - With Proofs
\else
Formation Path Following Control of Underactuated USVs
\fi
}

\author{{\AA}smund Eek$^{1}$, Kristin Y. Pettersen$^{1, 2}$, Else-Line M. Ruud$^{2}$ and Thomas R. Krogstad$^{2}$
\thanks{*This work was partly supported by the Research Council of Norway through the Centres of Excellence funding scheme, project No. 223254 – NTNU AMOS}
\thanks{$^{1}$Centre for Autonomous Marine Operations and Systems (NTNU AMOS), Department of Engineering Cybernetics, Norwegian University of Science and Technology, NO-7491 Trondheim}%
\thanks{$^{2}$Norwegian Defence Research Establishment (FFI), P.O. Box 25, N-2027
Kjeller, Norway}%
}

\begin{document}

\maketitle
\thispagestyle{empty}
\pagestyle{empty}

\begin{abstract}

This paper proposes a formation control method for two underactuated unmanned surface vessels (USVs) to follow curved paths in the presence of ocean currents. By uniting a line-of-sight (LOS) guidance law and the null-space-based behavioral control (NSB) framework, we achieve curved path following of the barycenter, while maintaining the desired vessel formation. 

The closed-loop dynamics are investigated using cascaded systems theory, and it is shown that the closed-loop system is USGES and UGAS, while the underactuated sway dynamics remains bounded. Both simulation and experimental results are presented to verify the theoretical results.




\end{abstract}

\begingroup
\setlength{\belowdisplayskip}{0pt}
\setlength{\abovedisplayskip}{0pt}
\setlength{\belowdisplayshortskip}{\belowdisplayskip}
\section{INTRODUCTION}

In recent years, the presence of autonomous vehicles has become more prominent, with self-driving cars being one of the most well-known fields. Another active area of research is within the maritime field, with a focus on both autonomous underwater and surface vessels. A significant advantage of autonomous systems is their ability to plan and execute tasks with reduced need for human interference. Current research areas include path planning and following, collision avoidance, and maneuvering in dynamic environments under the influence of disturbances such as wind, ocean currents, and waves for applications such as transportation, seafloor mapping, and in the oil and gas industry. Traditionally, these tasks have been performed using a single vessel. 
Some tasks, however, may be too complex to be solved by an individual system or be of such a nature that multiple systems are required to cooperate. With this, new challenges arise of how path following can be achieved while maintaining a desired overall formation.

The path following control problem for single underactuated marine vessels has been considered in several publications, see for instance \cite{FredriksenPettersen2006, Borhaug2008, Caharija2012, Belleter2019, Wiig2018, Oh2010, Pedro2003}. The line-of-sight (LOS) approach, by steering the vessel towards a point ahead on the path, is widely used to solve the path following problem, due to its intuitive structure and ability to counteract environmental disturbances, e.g. \cite{Borhaug2008, Belleter2019}. When considering the control problem for formations of marine vessels, several leader-follower approaches have been proposed where the follower adapts its speed and position relative to the leader to obtain the desired formation, e.g. \cite{Kyrkjebo2007, Lapierre2003, Breivik2008, Belleter2017}. However, leader-follower methods suffer from the fact that communication is unidirectional, meaning the leader will not adapt its speed to the follower.

The problem of straight-line path following formations of marine vessels is studied in \cite{Borhaug2006, Borhaug2011, Belleter2014}. The case without ocean currents is studied in \cite{Borhaug2006, Borhaug2011}, where the desired formation is obtained through each vessel in the formation using a LOS guidance law to follow the desired path, while the desired along-path distance between each vessel is obtained with a nonlinear velocity control law. Similarly, in \cite{Belleter2014}, the guidance law is extended with an integral LOS (ILOS) controller to counteract constant irrotational ocean currents. However, all of these methods are restricted to straight-line paths.

Another approach to the formation control problem is the null-space-based (NSB) behavioral control scheme, e.g. \cite{Arrichiello2006paper, Arrichiello2010, Pereda2011}. This centralized guidance system decomposes the control objective into different prioritized tasks, which are solved independently of each other using a closed loop inverse kinematics (CLIK) algorithm. The solutions of each task are then combined by projecting the solution of one task into the null-space of the higher-priority task. Expressing the control objectives in terms of fundamental tasks simplifies the control system design, as the tasks can be designed independently and then assembled to compose a complex behavior which could be difficult to create with a single objective function.



This paper aims to extend and combine the results for curved LOS path following for single USVs from \cite{Belleter2019} with the NSB framework presented in \cite{Arrichiello2006paper, Pereda2011} by replacing the barycenter CLIK control law with a LOS guidance law, to utilize the advantages of both methods. Contrary to \cite{Belleter2019}, where relative velocities are used, we will specify the LOS guidance law in terms of absolute velocity. This allows the implementation of the guidance law using only traditional sensors for estimating absolute velocities, such as an inertial measurement unit (IMU) and a global navigation satellite system (GNSS), which are available for most vessels, without the need of expensive sensors for measuring the relative velocities, which are less commonly available.

The closed-loop stability of the proposed LOS guidance law for path following of the barycenter, combined with the surge and heading autopilots from \cite{Moe2016} is analyzed using a cascaded system approach. Using the results from \cite{Pettersen2017}, we show that the closed-loop system is USGES and UGAS, while the underactuated sway dynamics are bounded. The LOS guidance law is then integrated into the NSB framework and the theoretical results are verified through both simulations and experiments. 

The paper is organized as follows. Section II gives a mathematical description of the unmanned surface vessel (USV) model. The control objectives are formalized in Section III, while the control system is presented in Section IV. In Section V we show that the closed-loop barycenter path following system is USGES, while the underactuated sway dynamics are proven to be bounded. Section VI and Section VII present the results of the guidance law from simulations and experiments respectively. Finally, Section VIII gives the conclusion of the work. 

\vspace*{-3mm}
\section{VESSEL MODEL}
\vspace*{-1mm}
\subsection{The Vessel Model}
\vspace*{-1mm}
The state of a marine surface vessel is given by the vector $\boldsymbol{\eta} \triangleq \left[x, y, \psi \right]^T$ which describes the position and orientation w.r.t the inertial frame $i$. The vector $\boldsymbol{\nu}\triangleq \left[u, v, r\right]^T$ contains the linear and angular velocities given in the body-fixed coordinate system $b$, where $u$ is the surge speed, $v$ the sway speed and $r$ the yaw rate. 

The ocean current velocity, denoted $\mathbf{V}_c$, expressed in the inertial frame, satisfies the following assumption:
\begin{assumption}\label{ass:vessel_modeling:eom:ocean_currents}
        The ocean current in the inertial frame is assumed to be constant and irrotational, i.e. $\mathbf{V}_c\triangleq\left[V_x, V_y, 0\right]^T$. Furthermore, there exists a constant $V_\text{max}>0$ such that $\lVert\mathbf{V}_c\rVert = \sqrt{V_x^2 + V_y^2} < V_\text{max}$, i.e. the ocean current is bounded.
\end{assumption}
Moreover, $\boldsymbol{\nu}_r \triangleq \boldsymbol{\nu} - \boldsymbol{\nu}_c$ is the relative velocity of the vessel, where $\mathbf{v}_c \triangleq\left[u_c, v_c, 0\right]^T$ is the ocean current velocity expressed in the body-fixed frame $b$ and obtained from $\mathbf{v}_c = \mathbf{R}(\psi)\mathbf{V}_c$, where $\mathbf{R}(\psi)$ is the rotation matrix from $b$ to $i$, defined as:
\begingroup
\setlength{\belowdisplayskip}{1pt}
\begin{equation}
        \mathbf{R}(\psi) = \begin{bmatrix}
                \cos(\psi) & -\sin(\psi)          & 0 \\
                \sin(\psi)          & \cos(\psi) & 0 \\
                0          & 0 & 1
                \end{bmatrix}.
\end{equation}
\endgroup
The kinematics and dynamics of the marine vessel are described by the 3-DOF maneuvering model \cite{Fossen2011,Borhaug2008,Moe2016}:
\vspace*{-3mm}
\begingroup
\setlength{\jot}{0.1ex}
\begin{subequations}
        \begin{align}
                \dot{\boldsymbol{\eta}} &= \mathbf{R}(\psi) \boldsymbol{\nu} \label{eq:vessel_modeling:eom:kinematic}\\
                \mathbf{M}_\text{RB}\dot{\boldsymbol{\nu}} + \mathbf{C}_\text{RB}(\boldsymbol{\nu})\boldsymbol{\nu} &= -\mathbf{M}_\text{A}\dot{\boldsymbol{\nu}}_r - \mathbf{C}_\text{A}(\boldsymbol{\nu}_r)\boldsymbol{\nu}_r \nonumber\\
                &\hphantom{{}=}- \mathbf{D}(\boldsymbol{\nu}_r) \boldsymbol{\nu}_r  + \mathbf{B}\mathbf{f}\label{eq:vessel_modeling:eom:kinetic}.
        \end{align}\label{eq:vessel_modeling:eom}%
\end{subequations}
\endgroup
The vector $\mathbf{f} \triangleq \left[T, \delta\right]^T$ is the control input vector, containing the surge thrust $T$ and the rudder angle $\delta$, respectively. The matrix $\mathbf{M}_\text{RB} = \mathbf{M}_\text{RB}^T>0$ is the rigid-body mass and inertia matrix, $\mathbf{C}_\text{RB}$ is the rigid-body Coriolis and centripetal matrix, $\mathbf{M}_\text{A} = \mathbf{M}_\text{A}^T>0$ is the hydrodynamic added mass matrix and $\mathbf{C}_\text{A}$ is the added mass Coriolis and centripetal matrix. Furthermore, $\mathbf{D}$ is the hydrodynamic damping matrix, and $\mathbf{B}\in\mathbb{R}^{3\times2}$ is the actuator configuration matrix.
\begin{assumption}\label{ass:model:symmetry}
        The USV is port-starboard symmetric.
\end{assumption}
\begin{assumption}\label{ass:model:xg}
        The body-fixed coordinate system is located at a distance $(x_g^*, 0)$ from the USV's center of gravity along the center-line of the vessel.
\end{assumption}
The matrices can be defined as
\begingroup
\setlength{\abovedisplayskip}{1pt}
\allowdisplaybreaks
\begin{gather}
        \mathbf{M}_x \triangleq \begin{bmatrix}
                m_{11}^x & 0         & 0 \\
                0          & m_{22}^x & m_{23}^x \\
                0          & m_{23}^x & m_{33}^x
                \end{bmatrix},\, 
        \mathbf{B} \triangleq \begin{bmatrix}
                b_{11} & 0  \\
                0          & b_{22} \\
                0          & b_{23}
                \end{bmatrix}\nonumber\\
        \mathbf{C}_x(\mathbf{z}) \triangleq \begin{bmatrix}
                        0 & 0         & -m_{22}^x z_2 - m_{23}^x z_3  \\
                        0          & 0 & m_{11} z_1 \\
                        m_{22}^x z_2 + m_{23}^x z_3        & -m_{11} z_1 & 0
                \end{bmatrix}\nonumber\\
        \mathbf{D}(\boldsymbol{\nu}_r) \triangleq \begin{bmatrix}
                d_{11} + d_{11}^q u_r & 0         & 0 \\
                0          & d_{22}^x & d_{23}^x \\
                0          & d_{32}^x & d_{33}^x
                \end{bmatrix},
\end{gather}
\endgroup
for $x\in\left\{RB, A\right\}$. The structure of $\mathbf{M}_x$ and $\mathbf{D}$ follows from \Crefrange{ass:model:symmetry}{ass:model:xg} and the structure of $\mathbf{C}_x$ is parametrized accordingly to \cite{Fossen2011}. To separate the sway-yaw subsystem, such that the yaw control does not affect the sway motion, the distance $x_g^*$ from Assumption \ref{ass:model:xg} is chosen such that $\mathbf{M}^{-1}\mathbf{B}\mathbf{M} = \left[\tau_u, 0, \tau_r\right]^T$, where $\mathbf{M} = \mathbf{M}_{RB} + \mathbf{M}_A$. Such a transformation does always exist for port-starboard symmetric vessels \cite{FredriksenPettersen2006}.

\subsection{Vessel Model in Component Form}
The model can be written in component form as
\begin{subequations}
        \begin{align}
                \dot{x} &= \cos(\psi)u - \sin(\psi)v\\
                \dot{y} &= \sin(\psi)u + \cos(\psi)v \\
                \dot{\psi} &= r \\
                \dot{u} &= - \frac{d_{11} + d_{11}^q u}{m_{11}} u + \frac{(m_{22}v + m_{23}r)}{m_{11}}r \nonumber\\
                &\hphantom{{}=} + \boldsymbol{\phi}_u^T(\psi, r)\boldsymbol{\theta}_u + \tau_u\\
                \dot{v} &= X(u, u_c)r + Y(u, u_c)v_r\label{EQ:NSB:ANALYSIS:VESSEL_MODEL:sway}\\
                \dot{r} &= F_r(u, v, r) + \boldsymbol{\phi}_r^T(u, v, r, \psi)\boldsymbol{\theta}_r + \tau_r,
        \end{align}\label{eq:model:component_form}%
\end{subequations}
where $m_{ji} \triangleq m_{ji}^{RB} + m_{ji}^A$ and $\boldsymbol{\theta}_u = \boldsymbol{\theta}_r = \left[V_x, V_y, V_x^2, V_y^2, V_x V_y\right]^T$ and the expressions for $\boldsymbol{\phi}_u^T(\psi, r)$, $X(u, u_c)$, $Y(u, u_c)$, $F_r(u, v, r)$ and $\boldsymbol{\phi}_r^T(u, v, r, \psi)$ are given in Appendix A. The functions $X(u, u_c)$ and $Y(u, u_c)$ are bounded for bounded inputs, and the following condition holds true for $Y(u, u_c)$:
\begin{assumption}\label{ass:model:Ymin}
        The function $Y(u, u_c)$ satisfies
        \begin{equation}
                Y(u, u_c) \leq - Y_\text{min} < 0,\quad \forall u\in \left[0, U_d\right].
        \end{equation}
\end{assumption}

\section{CONTROL OBJECTIVES}\label{sec:control_objectives}
The control objective is to make two underactuated USVs perform curved path following while aligning themselves such that the vector between them is perpendicular to the path with a desired inter-vessel distance. We propose a method where curved LOS path following is combined with the NSB framework, as it is proven in \cite{Belleter2019} that underactuated UAVs are able to counteract constant ocean currents with LOS guidance for curved paths. Moreover, the NSB framework simplifies the control system design by splitting the problem into fundamental tasks. Specifically, we choose the three following tasks, sorted by priority: collision avoidance, vessel formation and barycenter path following, where the first two tasks are defined according to \cite{Pereda2011}. 



To solve the third task, the objective of the control system is to make the barycenter of the two vessels converge to and follow a given smooth path $P$ while maintaining a desired total speed $U_d = \sqrt{u_d^2 + v^2}$ tangential to the path in the presence of unknown constant irrotational ocean currents. The path $P$ is parametrized using a path variable $\theta\in\mathbb{R}$ with respect to the inertial frame. Moreover, for each point on the path, $\left(x_p(\theta), y_p(\theta)\right)\in P$, a path tangential frame is introduced, see \Cref{fig:nsb:tasks:barycenter:problem_definition:path}. Using these definitions, the path following errors $\mathbf{p}_{pb}^p\triangleq[x_{pb}^p,y_{pb}^p]^T$ expressed in the path tangential frame is found to be 
\begin{equation}
        \begin{bmatrix}x_{pb}^p\\y_{pb}^p\end{bmatrix}\!=\!\begin{bmatrix}
    \cos\left(\gamma_p(\theta)\right) & -\sin\left(\gamma_p(\theta)\right) \\
    \sin\left(\gamma_p(\theta)\right) & \cos\left(\gamma_p(\theta)\right)
    \end{bmatrix}^T\!\begin{bmatrix}x_b - x_p(\theta)\\y_b - y_p(\theta)\end{bmatrix}, \label{eq:nsb:tasks:barycenter:problem_definition:path_following_errors}
\end{equation}
where $\gamma_p(\theta)$ is the path tangential angle. Hence, the task errors $x_{pb}^p$ and $y_{pb}^p$ express the position of the barycenter along the path frame tangential and orthogonal axis respectively. The barycenter path following objective is thus fulfilled if the trajectory of both vessels makes $x_{pb}^p$ and $y_{pb}^p$ converge to zero. 
\vspace*{-2.5mm}
\begin{figure}[htbp]
        \centering
        \begin{tikzpicture}[
            transform shape,
            dot/.style={circle,inner sep=1pt,fill},
            link/.style={distance=1.5mm,very thick, blue},
            boat/.style={distance=1.5mm,very thick, black},
            tangent/.style={
            decoration={
                markings,
                mark=
                    at position #1
                    with
                    {
                        \coordinate (tangent point-\pgfkeysvalueof{/pgf/decoration/mark info/sequence number}) at (0pt,0pt);
                        \coordinate (tangent unit vector-\pgfkeysvalueof{/pgf/decoration/mark info/sequence number}) at (1,0pt);
                        \coordinate (tangent orthogonal unit vector-\pgfkeysvalueof{/pgf/decoration/mark info/sequence number}) at (0pt,1);
                        \coordinate (tangent negative orthogonal unit vector-\pgfkeysvalueof{/pgf/decoration/mark info/sequence number}) at (0pt,-1);
                    }
            },
            postaction=decorate
        },
        use tangent/.style={
            shift=(tangent point-#1),
            x=(tangent unit vector-#1),
            y=(tangent orthogonal unit vector-#1)
        },
        use tangent/.default=1]
            
        \draw[ tangent=0.4, tangent=0.56, thick] plot [smooth] coordinates { (0,0) (2,1) (4,1) (5,1.5)};
        \coordinate[use tangent] (p) at (0,0);
        \coordinate (yaxis) at ($(p) + (0,1)$);
        \coordinate (xaxis) at ($(p) + (1,0)$);
        \coordinate[use tangent] (p_n) at (0,-1); 
        \coordinate (p_b) at (4,-0.5);
        \coordinate (p_b_pro) at ($(p)!(p_b)!(tangent negative orthogonal unit vector-1)$);
        
        \node[dot, label=$\mathbf{p}_b$] at (p_b) {};
        \node[label=$P$] at (5,1.5) {};
        \node[label={[xshift=1mm, yshift=-2mm]135:$\left(x_p(\theta)\text{,}\, y_p(\theta)\right)$}] at (p) {};

        \draw[dashed, -{Latex[length=2mm]}, thick] (p) -- (yaxis);
        \draw[dashed, -{Latex[length=2mm]}, thick] (p) -- (xaxis);
        \draw[-{Latex[length=2mm]}, use tangent, thick] (tangent point-1) -- node[at end, right] {$T$} (tangent unit vector-1);
        \draw[-{Latex[length=2mm]}, use tangent, thick] (tangent point-1) -- node[at end, right] {$N$} (0,-1);
        
        \draw[densely dashed, thick] (p) -- node [left, at end] {$y_{pb}^p$} (p_b_pro);
        \draw[densely dashed, thick] (p_b_pro) -- node [below] {$x_{pb}^p$}(p_b);

        \pic [draw, {Latex[length=1mm]}-, angle radius=6mm, angle eccentricity=1.2, right, "$\gamma_p(\theta)$", thick] {angle = tangent unit vector-1--tangent point-1--yaxis};
        
        \coordinate (origin) at (-1, -1.5);
        \coordinate (x) at (-1, 2);
        \coordinate (y) at (5, -1.5);
        \coordinate (x_b_pro) at ($(origin)!(p_b)!(x)$);
        \coordinate (y_b_pro) at ($(origin)!(p_b)!(y)$);
        \draw[-{Latex[length=3mm]}, thick] (origin) -- node [right, at end] {$X$} (x);
        \draw[-{Latex[length=3mm]}, thick] (origin) -- node [above, at end] {$Y$} (y);
        \draw[densely dotted, thick] (x_b_pro) -- node [above] {$x_b$}(p_b);
        \draw[densely dotted, thick] (y_b_pro) -- node [right] {$y_b$}(p_b);

        \end{tikzpicture}
        \caption{Definition of the path and barycenter path following errors.}
        \label{fig:nsb:tasks:barycenter:problem_definition:path}
    \end{figure}
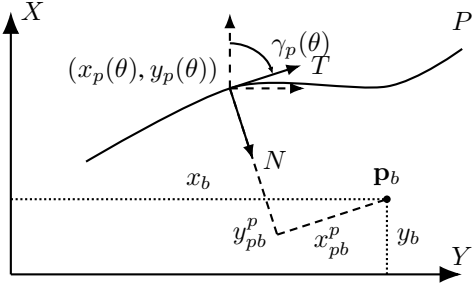


\vspace*{-5mm}
\section{CONTROL SYSTEM}
In this section, we first present the surge and yaw autopilots. Then, in \cref{subsec:control_system:nsb_tasks}, the NSB approach for generating the references for the autopilots is presented. Specifically, the high-priority NSB tasks of collision avoidance and vessel formation, and the corresponding CLIK algorithm and transformation to autopilot references, are presented. In \crefrange{subsec:control_system:barycenter_kinematics}{subsec:control_system:guidance_law}, the LOS guidance law for solving the third task, i.e. barycenter path following, is presented. The solutions of the barycenter path following task, $\mathbf{v}_{d, 3}$, are integrated into the NSB framework by projecting the task solutions onto the null-spaces of the higher-priority tasks, removing the components from lower-priority tasks that would conflict with the higher-priority tasks. The combined solutions from the NSB framework, $\mathbf{v}_\text{NSB}$, are then used to generate references for the autopilots.

\vspace*{-2mm}
\subsection{Surge and Yaw Controllers}
To control the surge and yaw states to their desired references, we will use the same autopilots as in \cite{Moe2016} due to their ocean current adaptation capabilities. Defining the error states
\begingroup
\setlength{\abovedisplayshortskip}{-1pt}
\setlength{\belowdisplayskip}{1pt}
\setlength{\jot}{0.1ex}
\begin{subequations}
    \begin{align}
        \tilde{u} &= u - u_d\\
        \tilde{\psi} &= \psi - \psi_d\\
        \dot{\tilde{\psi}} &= \dot{\psi} - \dot{\psi}_d\\
        \boldsymbol{\xi} &= \left[\tilde{u}, \tilde{\psi}, \dot{\tilde{\psi}}\right]^T,
    \end{align}\label{eq:nsb:analysis:error_states}%
\end{subequations}
\endgroup%
the following adaptive feedback linearizing PD-controller with sliding-mode is used to ensure tracking of the desired heading
\begingroup
\allowdisplaybreaks
\setlength{\jot}{0.1ex}
\begin{subequations}
    \begin{align}
        \tau_r &= -F_r(u, v, r) - \boldsymbol{\phi}_r^T(u, v, r, \psi)\hat{\boldsymbol{\theta}}_r + \ddot{\psi}_d\nonumber\\
        &- (k_\psi + \lambda k_r)\tilde{\psi} - (k_r + \lambda)\dot{\tilde{\psi}} - k_d\sign\left(\dot{\tilde{\psi}} + \lambda \tilde{\psi}\right)\\
        \dot{\hat{\boldsymbol{\theta}}}_r &= \gamma_r\boldsymbol{\phi}_r^T(u, v, r, \psi)\left(\dot{\tilde{\psi}} + \lambda \tilde{\psi}\right),
    \end{align}\label{eq:nsb:analysis:heading_controller}%
\end{subequations}
\endgroup
where the gains $k_\psi, k_r, \lambda, \gamma_r$ are constant and positive and the function $\sign(x)$ returns $1, 0$ and $-1$ when $x$ is positive, zero and negative, respectively. Further, a combined feedback linearizing and sliding-mode P-controller is used to track the desired surge speed
\begingroup
\allowdisplaybreaks
\setlength{\jot}{0.1ex}
\begin{subequations}
    \begin{align}
        \tau_u &= -\frac{1}{m_{11}}\left(m_{22}v + m_{23}r\right)r + \frac{d_{11}}{m_{11}}u_d - \boldsymbol{\phi}_u^T(\psi, r)\hat{\boldsymbol{\theta}}_u\nonumber\\
        &\hphantom{{}=}  + \dot{u}_d+ \frac{d_{11}^q}{m_{11}}u^2 - k_u\tilde{u} - k_e\sign(\tilde{u})\\
        \dot{\hat{\boldsymbol{\theta}}}_u &= \gamma_u\boldsymbol{\phi}_u^T(\psi, r)\tilde{u}.
    \end{align}\label{eq:nsb:analysis:surge_controller}%
\end{subequations}
\endgroup

\vspace*{-5mm}
\subsection{NSB Collision Avoidance and Vessel Formation Tasks}\label{subsec:control_system:nsb_tasks}
For each task, we define a task variable $\boldsymbol{\sigma}\in\mathbb{R}^m$:
\begin{equation}
    \boldsymbol{\sigma} = \mathbf{f}(\mathbf{p}), \label{eq:nsb:math:task_function}
\end{equation}
where $\mathbf{p} = [\mathbf{p}_1^T,\mathbf{p}_2^T]^T\in \mathbb{R}^{4}$ is the concatenated vector of system configurations, $\mathbf{p}_i \in\mathbb{R}^2$ is the position of vessel $i$, expressed in the inertial frame and $\mathbf{f}\colon\mathbb{R}^{4}\to\mathbb{R}^m$ is the task function which maps the system configuration to the task variable, see \cite{Arrichiello2006paper}. To track the desired task reference, $\boldsymbol{\sigma}_d(t)$, we use the CLIK algorithm presented in \cite{Arrichiello2006paper}
\begin{equation}
    \mathbf{v}_d = \mathbf{J}^\dagger \left(\dot{\boldsymbol{\sigma}}_d + \boldsymbol{\Lambda} \tilde{\boldsymbol{\sigma}}\right) \in \mathbb{R}^{2}, \label{eq:nsb:math:velocity_closed_loop}
\end{equation}
where $\tilde{\boldsymbol{\sigma}} = \boldsymbol{\sigma}_d - \boldsymbol{\sigma}$ is the task error, $\mathbf{J}\in\mathbb{R}^{m\times 4}$ is the configuration-dependent task Jacobian matrix, $(\cdot)^\dagger$ denotes the Moore-Penrose pseudoinverse and $\boldsymbol{\Lambda}\in \mathbb{R}^{m\times m}>0$ is a matrix of proportional gains. Now, let $\mathbf{v}_{d, i}$ denote the solution of the i\textsuperscript{th} priority task, given by \eqref{eq:nsb:math:velocity_closed_loop}. The velocities of each task are then combined by
\begin{equation}
    \mathbf{v}_\text{NSB} = \mathbf{v}_{d, 1} + \big(\mathbf{I} - \mathbf{J}_{1}^\dagger \mathbf{J}_{1}\big) \big[\mathbf{v}_{d,2} + \big(\mathbf{I} - \mathbf{J}_2^\dagger \mathbf{J}_2\big)\mathbf{v}_{d,3}\big], \label{eq:nsb:math:merging_tasks:projection}
\end{equation}
where $\mathbf{I}$ are the identity matrices of appropriate dimensions. The desired NSB velocity is decomposed into surge and yaw references by extending the method proposed in \cite{Arrichiello2006paper} with sideslip compensation, omitting vessel subscripts:
\begingroup
\setlength{\jot}{0.2ex}
\begin{align}
    u_d &= U_\text{NSB} \frac{1 + \cos\left(\chi_\text{NSB} - \chi\right)}{2} \label{eq:nsb:interface:converting:surge} \\
    \psi_d &= \chi_\text{NSB} - \underbrace{\arctan\left(\frac{v^b}{u_d}\right)}_{\beta_d}, \label{eq:nsb:interface:converting:heading}
\end{align}
\endgroup
where $U_\text{NSB}$ and $\chi_\text{NSB}$ are the norm and direction of $\mathbf{v}_\text{NSB}$, respectively, and $\chi$ is the course of the vessel. The second term of \eqref{eq:nsb:interface:converting:heading} is the desired sideslip angle for each vessel, to make the vessel's course parallel to the path when the vessel's sway speed is non-zero.

Specifically, the collision avoidance task for the i\textsuperscript{th} vessel is defined as 
\begingroup
\setlength{\abovedisplayshortskip}{-1pt}
\begin{equation}
    \sigma_{ca} = \lVert\mathbf{p}_i - \mathbf{p}_o\rVert\in\mathbb{R}, \label{eq:control_objectives:collision_avoidance_objective}
\end{equation}
\endgroup
where $\mathbf{p}_i$, $\mathbf{p}_o\in\mathbb{R}^2$ are the positions of the i\textsuperscript{th} vessel and the other vesssel, respectively, expressed in the inertial frame. The task is only activated when the inter-vessel distance is below a certain threshold, i.e. $\sigma_{ca} < \sigma_{ca, d}$. As \eqref{eq:control_objectives:collision_avoidance_objective} is scalar, the collision avoidance gain in \eqref{eq:nsb:math:velocity_closed_loop} is reduced to $\lambda_{ca}\in\mathbb{R}$. Moreover, the vessel formation task is defined as
\begin{equation}
    \boldsymbol{\sigma}_f = \lVert \mathbf{p}_1 - \mathbf{p}_b\rVert\in\mathbb{R}^2,\label{eq:control_objectives:formation_objective}
\end{equation}
where $\mathbf{p}_1$, $\mathbf{p}_b = [x_b, y_b]^T$ are the positions of the first vessel and the barycenter position, to be defined later, respectively, expressed in the inertial frame. To fulfill the vessel formation control objective, the desired task function value, $\boldsymbol{\sigma}_{f, d}^p$, is expressed in the path tangential frame and transformed to the inertial frame by: 
\begin{equation}
    \boldsymbol{\sigma}_{f, d} = \mathbf{R}(\gamma_p(\theta))^T \boldsymbol{\sigma}_{f, d}^p.\label{eq:control_objectives:desired_formation_objective}
\end{equation}
To be able to specify the weighting for along- and cross-track formation task errors independently, we express the vessel formation gain in terms of $\boldsymbol{\Lambda}_f^p$ which is transformed similarly to \eqref{eq:control_objectives:desired_formation_objective} to obtain the vessel formation gain $\boldsymbol{\Lambda}_f$.

\vspace*{-2mm}
\subsection{Barycenter Kinematics}\label{subsec:control_system:barycenter_kinematics}
The barycenter given the two vessel positions can be expressed as:
\begin{equation}
    \boldsymbol{\sigma}_b = \mathbf{p}_b = \frac{1}{2}\left(\mathbf{p}_1 + \mathbf{p}_2\right). \label{eq:nsb:tasks:barycenter:kinematics:barycenter}
\end{equation}
Next, as the position of the barycenter cannot be controlled directly, only through each of the vessels, the barycenter kinematics is expressed in terms of the kinematics of each vessel \eqref{eq:vessel_modeling:eom:kinematic} by taking the time derivative of \eqref{eq:nsb:tasks:barycenter:kinematics:barycenter}
\begingroup
\setlength{\belowdisplayskip}{1pt}
\begin{subequations}
    \begin{align}
        \dot{x}_b &= \frac{1}{2}\Big[u_1\!\cos\psi_1\! -\! v_1\! \sin\psi_1\! +\! u_2\!\cos\psi_2\! -\! v_2\!\sin\psi_2\Big] \\
        \dot{y}_b &= \frac{1}{2}\Big[u_1\!\sin\psi_1\! +\! v_1\! \cos\psi_1\! +\! u_2\!\sin\psi_2\! +\! v_2\!\cos\psi_2\Big].
    \end{align}\label{eq:nsb:tasks:barycenter:kinematics:barycenter_kinematics_component_form}%
\end{subequations}
\endgroup
The path following error dynamics is then computed by substituting \eqref{eq:nsb:tasks:barycenter:kinematics:barycenter_kinematics_component_form} into the time derivative of \eqref{eq:nsb:tasks:barycenter:problem_definition:path_following_errors}
\begin{subequations}
    \begin{align}
        \dot{x}_{pb}^p\!&=\!\frac{1}{2}U_1\!\!\cos\!\left(\chi_1\!\!-\!\gamma_p\!\right)\!\!+\!\!\frac{1}{2}U_2\!\cos\!\left(\chi_2\!\!-\!\gamma_p\!\right)\!\!-\!\!\dot{\theta}(1\!\!-\!\kappa(\theta)y_{pb}^p\!\!\label{EQ:NSB:TASKS:BARYCENTER:KINEMATICS:PATH_FOLLOWING_ERRORDYNAMICS:x}\\
        \dot{y}_{pb}^p\!&=\!\frac{1}{2}U_1\!\sin\!\left(\chi_1\!-\!\gamma_p\!\right)\!+\!\frac{1}{2}U_2\!\sin\!\left(\chi_2\!-\!\gamma_p\!\right)\!-\!\kappa(\theta)\dot{\theta}x_{pb}^p,\!\label{EQ:NSB:TASKS:BARYCENTER:KINEMATICS:PATH_FOLLOWING_ERRORDYNAMICS:y}
    \end{align}\label{EQ:NSB:TASKS:BARYCENTER:KINEMATICS:PATH_FOLLOWING_ERRORDYNAMICS}%
\end{subequations}
where $\kappa(\theta)$ is the curvature of $P$ at $\theta$ and $\chi_i$ the course of vessel $i$. 

\vspace*{-2mm}
\subsection{Path Parametrization}\label{subsec:control_system:path_parametrization}
As the path is parametrized by the path variable $\theta$, it is possible to use the update law of the path variable as an extra degree of freedom when designing the controller \cite{Lapierre2007}. As in \cite{Belleter2019}, where the update law is chosen to obtain a desirable behavior of the $x_{pb}^p$ dynamic, a similar approach will be used here where the update law is chosen such that the propagation speed of the path tangential frame cancel the undesirable terms of \eqref{EQ:NSB:TASKS:BARYCENTER:KINEMATICS:PATH_FOLLOWING_ERRORDYNAMICS:x}:
\begin{equation}
    \dot{\theta}\! =\!\frac{1}{2}U_1\!\cos\!\left(\chi_1\!-\!\gamma_p\!\right)\!+\!\frac{1}{2}U_2\!\cos\!\left(\chi_2\!-\!\gamma_p\!\right)\!+\! k_\theta f_\theta\!\left(x_{pb}^p,\! y_{pb}^p\!\right)\!,\!\!\label{eq:nsb:tasks:barycenter:path:update_law}
\end{equation}
where $k_\theta \in\mathbb{R}_{>0}$ is a control gain and $f_\theta\colon\mathbb{R}^{2}\to\mathbb{R}$ is a function satisfying $f_\theta(x_{pb}^p, y_{pb}^p)x_{pb}^p>0$ that we use to ensure a desirable along-track error dynamics. Specifically, we choose:
\begin{equation}
    f_\theta\left(x_{pb}^p, y_{pb}^p\right) = \frac{x_{pb}^p}{\sqrt{1 + \left(x_{pb}^p\right)^2}}.
\end{equation}
Inserting \eqref{eq:nsb:tasks:barycenter:path:update_law} into \eqref{EQ:NSB:TASKS:BARYCENTER:KINEMATICS:PATH_FOLLOWING_ERRORDYNAMICS:x} we obtain the following along-track error dynamics
\begin{equation}
    \dot{x}_{pb}^p = - k_\theta \frac{x_{pb}^p}{\sqrt{1 + \left(x_{pb}^p\right)^2}} + \dot{\theta}\kappa(\theta)y_{pb}^p, \label{eq:nsb:tasks:barycenter:path:inserted_update_law_along_track_error_dynamics}
\end{equation}
where the choice of $f_\theta$ introduces a stabilizing term in the along-track error dynamics. 

\subsection{Guidance Law}\label{subsec:control_system:guidance_law}
To obtain path following for the barycenter, we choose the following LOS guidance law:
\begin{equation}
    \chi_{b,d} = \gamma_p(\theta) -  \arctan\left(\frac{y_{pb}^p}{\Delta\left(\mathbf{p}_{pb}^p\right)}\right). \label{eq:nsb:tasks:barycenter:guidance:guidance_law_heading}
\end{equation}
The guidance law consists of two terms. The first term is a feedforward term of the path tangential angle and the last term is a traditional line-of-sight term for steering the barycenter towards the desired path, see \cref{fig:nsb:tasks:barycenter:guidance:illustration_of_los_guidance}. Contrary to the LOS guidance law of \cite{Belleter2019}, the ocean current observer and the ocean current dependent term $g$ used to compensate for the ocean current is not present in  \eqref{eq:nsb:tasks:barycenter:guidance:guidance_law_heading} as the ocean current compensation is instead handled by the adaptive autopilots \eqref{eq:nsb:analysis:heading_controller}--\eqref{eq:nsb:analysis:surge_controller}. The solutions of the LOS guidance law are then integrated into the NSB framework by defining the desired barycenter task velocity
\begingroup
\setlength{\belowdisplayskip}{2pt}
\begin{equation}
    \mathbf{v}_{d, 3} = U_d\begin{bmatrix}
        \cos\chi_{b,d} \\
        \sin\chi_{b,d}
    \end{bmatrix}.
\end{equation}
\endgroup
Inspired by \cite{Belleter2019}, the lookahead term $\Delta(\mathbf{p}_{pb}^p)$ is chosen to have one constant part and one part depending on the path following errors
\begin{equation}
    \Delta\left(\mathbf{p}_{pb}^p\right) = \sqrt{\mu + \left(x_{pb}^p\right)^2 + \left(y_{pb}^p\right)^2}, \label{EQ:NSB:TASKS:BARYCENTER:GUIDANCE:LOOKAHEAD}
\end{equation}
where $\mu \in \mathbb{R}_{>0}$ is a constant. Contrary to \cite{Belleter2019}, the lookahead distance is not required to depend on $y_{pb}^p$ for the conditions of \cref{LEM:NSB:ANALYSIS:BARYCENTER:BOUNDEDNESS_X2} to hold. However, we still choose to include the dependence on the cross-track error to obtain a greater lookahead distance when the barycenter is far away from the desired path. 

Substituting \eqref{eq:nsb:tasks:barycenter:guidance:guidance_law_heading} in \eqref{EQ:NSB:TASKS:BARYCENTER:KINEMATICS:PATH_FOLLOWING_ERRORDYNAMICS:y} we obtain the following cross-track error dynamics
\begin{align}
    \dot{y}_{pb}^p &= \frac{1}{2} U_{d,1} \sin\left(\psi_{d,1} + \tilde{\psi}_1 + \beta_{d,1} - \gamma_p\right)\nonumber\\
    &\hphantom{{}=} + \frac{1}{2} U_{d,2} \sin\left(\psi_{d,2} + \tilde{\psi}_2 + \beta_{d,2} - \gamma_p\right)\nonumber\\
    &\hphantom{{}=} - \kappa(\theta)\dot{\theta}x_{pb}^p\nonumber\\
    &\hphantom{{}=} + \frac{1}{2}\tilde{u}_1\sin\left(\psi_1 - \gamma_p\right) + \frac{1}{2}\tilde{u}_2\sin\left(\psi_2 - \gamma_p\right) \label{eq:nsb:tasks:barycenter:guidance:cross_track_error_dynamics_before_inserting_guidance_law}\\
    &= -\frac{1}{2}\left(U_{d,1} + U_{d,2}\right)\frac{y_{pb}^p}{\sqrt{\Delta^2 + \left(y_{pb}^p\right)^2}} - \kappa(\theta)\dot{\theta}x_{pb}^p \nonumber\\
    &\hphantom{{}=}\!+\! G_1\!\left(\tilde{\psi}_1,\tilde{u}_1,\psi_{d,1},U_{d,1},\tilde{\psi}_2,\tilde{u}_2,\psi_{d,2},U_{d,2},y_{pb}^p\!\right)\!,\!\!\label{EQ:NSB:TASKS:BARYCENTER:GUIDANCE:CROSS_TRACK_ERROR_DYNAMICS}
\end{align}
where $U_{d,i} = \sqrt{u_{d,i}^2 + v_i^2}$ is the total desired speed of vessel $i$ and $G_1(\cdot)$ is a perturbing term of the vessels autopilots' error states, given by
\begin{equation}
    G_1(\cdot) = \frac{1}{2} \sum_{i=1}^2 G_2\left(\tilde{\psi}_i, \tilde{u}_i, \psi_{d,i}, U_{d,i}, y_{pb}^p\right)
\end{equation}
with 
\begingroup 
\allowdisplaybreaks
\begin{align}
    G_2&(\tilde{\psi}, \tilde{u}, \psi_{d}, U_{d}, y_{pb}^p) = \tilde{u}\sin\left(\psi - \gamma_p\right)\nonumber\\
    &\hphantom{{}=}+ U_d\left(1 - \cos\tilde{\psi}\right)\sin\left(\arctan\left(\frac{y_{pb}^p}{\Delta}\right)\right)\nonumber\\
    &\hphantom{{}=} + U_d\cos\left(\arctan\left(\frac{y_{pb}^p}{\Delta}\right)\right)\sin\tilde{\psi}.
\end{align}
\endgroup
Note that $G_1(\cdot)$ satisfy
\begin{subequations}
    \begin{gather}
        G_1\left(0, 0, \psi_{d,1}, U_{d,1}, 0, 0, \psi_{d,2}, U_{d,2}, y_{pb}^p\right) = 0 \label{eq:nsb:tasks:barycenter:guidance:G_1_features:zero}\\
        \lVert G_1(\cdot)\rVert \leq \zeta_1\left(U_{d,1}, U_{d,2}\right) \big \lVert \big[\tilde{\psi}_1, \tilde{u}_1, \tilde{\psi}_2, \tilde{u}_2\big]^T \big \rVert, \label{eq:nsb:tasks:barycenter:guidance:G_1_features:linear}
    \end{gather}\label{eq:nsb:tasks:barycenter:guidance:G_1_features}%
\end{subequations}
where $\zeta_1\left(U_{d,1}, U_{d,2}\right)>0$. This shows that the perturbing term $G_1(\cdot)$ is zero when the perturbing states are zero, and has at most linear growth in the perturbing states. 

The desired yaw rate is found by substituting \eqref{eq:nsb:tasks:barycenter:guidance:guidance_law_heading} into \eqref{eq:nsb:interface:converting:heading} and taking the time-derivative:
\begingroup
\allowdisplaybreaks
\begin{align}
    r_d &= \dot{\psi}_d = \kappa(\theta)\dot{\theta} - \frac{\dot{v}u_d - \dot{u}_d v}{u_d^2 + v^2}\nonumber\\
    \!-&\frac{1}{\Delta^2\!+\!\left(y_{pb}^p\right)^2}\!\left[\Delta\dot{y}_{pb}^p\!\!-y_{pb}^p\!\left(\frac{\partial\Delta}{\partial x_{pb}^p}\dot{x}_{pb}^p\!+\!\frac{\partial\Delta}{\partial y_{pb}^p}\dot{y}_{pb}^p\!\right)\!\right].\!\!
\end{align}
\endgroup
Substituting the along-track and cross-track error dynamics \eqref{eq:nsb:tasks:barycenter:path:inserted_update_law_along_track_error_dynamics} and \eqref{EQ:NSB:TASKS:BARYCENTER:GUIDANCE:CROSS_TRACK_ERROR_DYNAMICS} along with the sway dynamics \eqref{EQ:NSB:ANALYSIS:VESSEL_MODEL:sway} we obtain
\begingroup 
\setlength{\abovedisplayskip}{1pt}
\allowdisplaybreaks
\begin{align}
    r_d &= \kappa(\theta)\dot{\theta} - \frac{u_d}{u_d^2 + v^2}(X(u, u_c) r + Y(u,u_c)v \nonumber\\
    &\hphantom{{}=}- Y(u,u_c)v_c) + \frac{\dot{u}_d v}{u_d^2 + v^2}\nonumber\\
    &\hphantom{{}=}- \frac{1}{\Delta^2 + \left(y_{pb}^p\right)^2} \Bigg[\left(\Delta - y_{pb}^p\frac{\partial \Delta}{\partial y_{pb}^p}\right)\Bigg[G_1(\cdot)\nonumber\\
    &\hphantom{{}=}- \frac{1}{2}\left(U_{d,1} + U_{d,2}\right)\frac{y_{pb}^p}{\sqrt{\Delta^2 + \left(y_{pb}^p\right)^2}} - \kappa(\theta)\dot{\theta}x_{pb}^p\Bigg]\nonumber\\
    &\hphantom{{}=}- y_{pb}^p\frac{\partial \Delta}{\partial x_{pb}^p}\Bigg[- k_\theta \frac{x_{pb}^p}{\sqrt{1 + \left(x_{pb}^p\right)^2}} + \dot{\theta}\kappa(\theta)y_{pb}^p\Bigg]\Bigg].\label{eq:nsb:analysis:barycenter:r_d}
\end{align}
\endgroup
\begin{remark}
    Looking at \eqref{eq:nsb:analysis:barycenter:r_d} it is clear how the expression of $\dot{\psi}_d$ contains terms depending on $\dot{v}$, which depends on the unknown ocean current $u_c$ and the relative surge and sway speeds $u_r$ and $v_r$. Interestingly, contrary to \cite{Belleter2019} where $\dot{\psi}_d$ was also dependent on unknown variables and could not be realised, this is not an issue in \eqref{eq:nsb:analysis:barycenter:r_d}. Specifically, in \cite{Belleter2019}, the terms depending on the unknown ocean currents appeared through the along- and cross-track error dynamics. As these states could not be measured directly, they had to be calculated, requiring knowledge about the unknown ocean current. However, in \eqref{eq:nsb:analysis:barycenter:r_d}, the terms depending on the unknown variables appear only through the sway dynamics $\dot{v}$ of the vessel, which can be measured using e.g. a IMU, and is therefore available for feedback. This, comes from our choice of defining the LOS guidance law in terms of absolute velocities, contrary to \cite{Belleter2019} where relative velocities are used. Although this apparently is a small difference, an important implication is that this allows \eqref{eq:nsb:analysis:barycenter:r_d} to be realised by obtaining $\dot{v}$ through available sensor measurement instead of requiring knowledge about the unknown ocean current as $\dot{\psi}_d$ in \cite[Eq. (35)]{Belleter2019}.
\end{remark}

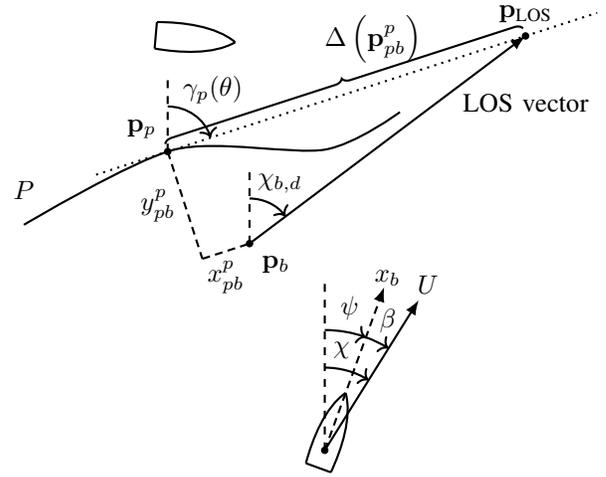
\begin{figure}[htbp]
        \centering
        \newcommand\vessel[4]{%
            \coordinate (vessel_base_#1) at (#2, #3);
            \pgfmathsetmacro\bowx{cos(#4)}
            \pgfmathsetmacro\bowy{sin(#4)}
            \pgfmathsetmacro\basexport{cos(#4 + 45)}
            \pgfmathsetmacro\baseyport{sin(#4 + 45)}
            \pgfmathsetmacro\basexstb{cos(#4 - 45)}
            \pgfmathsetmacro\baseystb{sin(#4 - 45)}
            \coordinate (vessel_aft_#1) at ($(vessel_base_#1) - (1.4142 * 0.3 * \bowx,1.4142 * 0.3 * \bowy)$);
            \coordinate (vessel_bow_#1) at ($(vessel_aft_#1) + (1.4142 * 0.875 * \bowx,1.4142 * 0.875 * \bowy)$);
            \draw[thick] ($(vessel_aft_#1) + (0.25 * \basexport,0.25 * \baseyport)$) -- ($(vessel_aft_#1) + (0.25 * \basexstb,0.25 * \baseystb)$);
            \pgfmathsetmacro\outtop{#4 - 5}
            \pgfmathsetmacro\outbottom{#4 + 5}
            \pgfmathsetmacro\intop{#4 + 215}
            \pgfmathsetmacro\inbottom{#4 + 145}
            \draw[thick] ($(vessel_aft_#1) + (0.25 * \basexport,0.25 * \baseyport)$) to [out=\outbottom,in=\inbottom,looseness=0.8] (vessel_bow_#1);
            \draw[thick] ($(vessel_aft_#1) + (0.25 * \basexstb,0.25 * \baseystb)$) to [out=\outtop,in=\intop,looseness=0.8] (vessel_bow_#1);
    }
        \begin{tikzpicture}[
            transform shape,
            extended line/.style={shorten >=-#1,shorten <=-#1},
            extended line/.default=1cm,
            one end extended/.style={shorten >=-#1},
            one end extended/.default=1cm,
            dot/.style={circle,inner sep=1pt,fill},
            link/.style={distance=1.5mm,very thick, blue},
            boat/.style={distance=1.5mm,very thick, black},
            tangent/.style={
            decoration={
                markings,
                mark=
                    at position #1
                    with
                    {
                        \coordinate (tangent point-\pgfkeysvalueof{/pgf/decoration/mark info/sequence number}) at (0pt,0pt);
                        \coordinate (tangent unit vector-\pgfkeysvalueof{/pgf/decoration/mark info/sequence number}) at (1,0pt);
                        \coordinate (tangent far away vector-\pgfkeysvalueof{/pgf/decoration/mark info/sequence number}) at (5,0pt);
                        \coordinate (tangent orthogonal unit vector-\pgfkeysvalueof{/pgf/decoration/mark info/sequence number}) at (0pt,1);
                        \coordinate (tangent negative orthogonal unit vector-\pgfkeysvalueof{/pgf/decoration/mark info/sequence number}) at (0pt,-1);
                    }
            },
            postaction=decorate
        },
        use tangent/.style={
            shift=(tangent point-#1),
            x=(tangent unit vector-#1),
            y=(tangent orthogonal unit vector-#1)
        },
        use tangent/.default=1]
        |   
            \draw[ tangent=0.4, tangent=0.56, thick] plot [smooth] coordinates { (0,0) (2,1) (4,1) (5,1.5)};
            \node[label=$P$] at (0,0.1) {};
            \coordinate[use tangent] (p_p) at (0,0); 
            \node[dot, label=120:$\mathbf{p}_p$] at (p_p) {};
            \draw[dotted,  extended line=1cm, thick] (tangent point-1) -- (tangent far away vector-1);
            \draw[decorate, decoration={brace,raise=3pt,amplitude=3pt}, thick] (tangent point-1)  -- node[above, yshift=8pt, xshift=10pt]{$\Delta\left(\mathbf{p}_{pb}^p\right)$} (tangent far away vector-1) ;
            \coordinate (p_los) at (tangent far away vector-1);
            \node[dot, label=90:$\mathbf{p}_\text{LOS}$] at (p_los) {};
            
            \vessel{1}{2}{2.5}{-5};
            \vessel{2}{4}{-3}{70};
            \node[dot] at (vessel_base_2) {};
            
            \coordinate (p_b) at ($(vessel_base_1)!0.5!(vessel_base_2)$);
            \node[dot, label=-45:$\mathbf{p}_b$] at (p_b) {};
            
            \coordinate (p_b_pro) at ($(p_p)!(p_b)!(tangent negative orthogonal unit vector-1)$);
            \draw[densely dashed, thick] (p_p) -- node [left] {$y_{pb}^p$} (p_b_pro);
            \draw[densely dashed, thick] (p_b_pro) -- node [below] {$x_{pb}^p$}(p_b);
            
            \coordinate (yaxis_vessel) at ($(vessel_base_2) + (0,2.3)$);
            \coordinate (yaxis_p_b) at ($(p_b) + (0,1)$);
            \coordinate (yaxis_p_p) at ($(p_p) + (0,1)$);
            \draw[dashed,  thick] (vessel_base_2) -- (yaxis_vessel);
            \draw[dashed,  thick] (p_b) -- (yaxis_p_b);
            \draw[dashed, thick] (p_p) -- (yaxis_p_p);
            
            \draw[-{Latex[length=2mm]}, densely dashed, thick, one end extended=1.5cm] (vessel_base_2) -- node[pos=3] {$x_b$} (vessel_bow_2);
            
            \coordinate (sog) at ($(vessel_base_2) + (1.25, 2)$);
            \draw[-{Latex[length=2mm]}, thick] (vessel_base_2) -- node[pos=1.1] {$U$} (sog);
            
            \draw[-{Latex[length=2mm]}, thick] (p_b) {} -- node[at end, label={[label distance=0.5cm]below:\text{LOS vector}}] {} (p_los);

            \pic [draw, <-, thick, angle radius=16mm, angle eccentricity=1.2, "$\psi$"] {angle = vessel_bow_2--vessel_base_2--yaxis_vessel};
            \pic [draw, <-, thick, angle radius=16mm, angle eccentricity=1.2, "$\beta$"] {angle = sog--vessel_base_2--vessel_bow_2};
            \pic [draw, <-, thick, angle radius=11mm, angle eccentricity=1.2, "$\chi$", pic text options={shift={(-4pt,2pt)}}] {angle = sog--vessel_base_2--yaxis_vessel};
            \pic [draw, <-, thick, angle radius=6mm, angle eccentricity=1.5, "$\chi_{b,d}$"] {angle = p_los--p_b--yaxis_p_b};
            
            \coordinate (intersection_tangent_yaxis) at (intersection of vessel_base_2--yaxis_vessel) and tangent point-1--tangent far away vector-1);
            \pic [draw,<-, angle radius=6mm, angle eccentricity=1.7, "$\gamma_p(\theta)$", thick] {angle = p_los--p_p--yaxis_p_p};

          \end{tikzpicture}
        \caption[Illustration of the LOS guidance law for path following for the barycenter.]{Illustration of the LOS guidance law for path following for the barycenter. The subscripts for the lower vessel's states are omitted for simplicity.}
        \label{fig:nsb:tasks:barycenter:guidance:illustration_of_los_guidance}
        \vspace*{-2.5mm}
    \end{figure}

\section{CLOSED-LOOP ANALYSIS}
In this section, we will analyse the closed-loop error dynamics of the barycenter path following task, described in \crefrange{subsec:control_system:barycenter_kinematics}{subsec:control_system:guidance_law}

The closed-loop error variables are defined as follows:
\begingroup
\setlength{\jot}{0.25ex}
\begin{subequations}
    \begin{align}
        \tilde{\mathbf{X}}_1 &\triangleq \left[x_{pb}^p, y_{pb}^p\right]^T\\
        \tilde{\mathbf{X}}_{2,i} &\triangleq \left[\tilde{u}_i, \dot{\tilde{\psi}}_i, s_i\right]^T\\
        \tilde{\mathbf{X}}_2 &\triangleq \left[\tilde{\mathbf{X}}_{2,1}^T, \tilde{\mathbf{X}}_{2,2}^T\right]^T,
    \end{align}\label{eq:nsb:analysis:barycenter:error_states}%
\end{subequations}%
\endgroup%
where $\tilde{\mathbf{X}}_{2,i}$, contains the autopilot error states of each vessel, that converge independent of $\tilde{\mathbf{X}}_1$, and the coordinate transformation $s_i = \tilde{\psi}_i + \lambda \dot{\tilde{\psi}}_i$ is applied motivated by \cite{Moe2016}. Moreover, we define the estimation errors $\tilde{\boldsymbol{\theta}}_x = \hat{\boldsymbol{\theta}}_x - \boldsymbol{\theta}_x$, where $x \in \{u, r\}$, for each vessel respectively.  

Thus, the error dynamics of the closed-loop barycenter path following system consisting of the vessels, given by \eqref{eq:model:component_form}, and the control laws \eqref{eq:nsb:analysis:heading_controller}--\eqref{eq:nsb:analysis:surge_controller} with the guidance laws \eqref{eq:nsb:tasks:barycenter:guidance:guidance_law_heading} may be written as 
\begin{subequations}
    \begin{align}
        \dot{\tilde{\boldsymbol{X}}}_1\!&=\!\!\begin{bmatrix}
            - k_\theta \frac{x_{pb}^p}{\sqrt{1 + \left(x_{pb}^p\right)^2}} + \dot{\theta}\kappa(\theta)y_{pb}^p \\
            \frac{-\frac{1}{2}\left(U_{d,1}+U_{d,2}\right) y_{pb}^p}{\sqrt{\Delta^2+\left(y_{pb}^p\right)^2}}-\kappa(\theta)\dot{\theta}x_{pb}^p
        \end{bmatrix}+\begin{bmatrix}
        0\\
        G_1(\cdot)
        \end{bmatrix}\label{eq:nsb:analysis:barycenter:closed_loop:X1}\\
        \dot{\tilde{\boldsymbol{X}}}_{2,i}\!&=\!\!\begin{bmatrix}
        -\!\left(\!\frac{d_{11}}{m_{11}}\!+\! k_{u,i}\!\right)\!\tilde{u}_i\!-\!\boldsymbol{\phi}_u^T\!(\cdot)\tilde{\boldsymbol{\theta}}_{u,i}\!-\! k_{e,i}\sign(\tilde{u}_i)\!\\
        -\lambda_i\tilde{\psi}_i+s\!\\
        -k_{\psi,i}\tilde{\psi}_i\!-\! k_{r,i}s_i\!-\!\boldsymbol{\phi}_r^T\!(\cdot)\tilde{\boldsymbol{\theta}}_{r,i}\!-\! k_{d,i}\sign(s_i)\!
        \end{bmatrix}\!\label{eq:nsb:analysis:barycenter:closed_loop:X2}\\
        \dot{\tilde{\boldsymbol{\theta}}}_{r,i} &= \gamma_r\boldsymbol{\phi}_r^T(u, v, r, \psi)s_i\label{eq:nsb:analysis:barycenter:closed_loop:theta_r}\\
        \dot{\tilde{\boldsymbol{\theta}}}_{u,i} &= \gamma_u\boldsymbol{\phi}_u^T(\psi, r)\tilde{u}_i\label{eq:nsb:analysis:barycenter:closed_loop:theta_u}\\
        \dot{v}_i &= X(u_{d,i} + \tilde{u}_i, u_c)r_{d_i} + X(u_{d,i} + \tilde{u}_i, u_c)\tilde{r}_i \nonumber\\
        &\hphantom{{}=}+ Y(u_{d,i} + \tilde{u}_i, u_c)v_i - Y(u_{d,i} + \tilde{u}_i, u_c)v_c.\label{eq:nsb:analysis:barycenter:closed_loop:sway}
    \end{align}\label{eq:nsb:analysis:barycenter:closed_loop}%
\end{subequations}
To solve the barycenter control objective defined in \cref{sec:control_objectives}, the error states $\tilde{\boldsymbol{X}}_1$ and $\tilde{\boldsymbol{X}}_2$ should converge to zero, while the estimation errors $\tilde{\boldsymbol{\theta}}_{r,i}$, $\tilde{\boldsymbol{\theta}}_{u,i}$ and the sway velocity $v_i$ should remain bounded. 
\begin{lemma}[Forward Completeness]\label{LEM:NSB:ANALYSIS:BARYCENTER:FORWARD_COMPLETENESS}
The trajectories of the closed-loop system \eqref{eq:nsb:analysis:barycenter:closed_loop} are forward complete
\end{lemma}
\begin{proof}
\ifincludeproofs
The proof of this lemma is given in Appendix B.
\fi
\end{proof}
\begin{lemma}[Boundedness near $(\tilde{\mathbf{X}}_1, \tilde{\mathbf{X}}_2) = \mathbf{0}$]\label{LEM:NSB:ANALYSIS:BARYCENTER:BOUNDEDNESS_X1_X2}
The system \eqref{eq:nsb:analysis:barycenter:closed_loop:sway} is bounded near the manifold $(\tilde{\mathbf{X}}_1, \tilde{\mathbf{X}}_2) = \mathbf{0}$ if and only if the curvature of $P$ satisfies the following condition:
\begin{equation}
    \kappa_\text{max} \triangleq \max_{\theta \in P}|\kappa(\theta)| < \frac{Y_\text{min}}{ X_\text{max}}, \quad X_\text{max} \triangleq |X(u, u_c)|_\infty. \label{eq:nsb:analysis:barycenter:boundedness_X1_X2:max_curvature}
\end{equation}
\end{lemma}
\begin{proof}
\ifincludeproofs
The proof of this lemma is given in Appendix B.
\fi
\end{proof}
\begin{lemma}[Boundedness near $\tilde{\mathbf{X}}_2 = \mathbf{0}$]\label{LEM:NSB:ANALYSIS:BARYCENTER:BOUNDEDNESS_X2}
The system \eqref{eq:nsb:analysis:barycenter:closed_loop:sway} is bounded near the manifold $\tilde{\mathbf{X}}_2 = \mathbf{0}$, independently of $\tilde{\mathbf{X}}_1$, if the conditions of \cref{LEM:NSB:ANALYSIS:BARYCENTER:BOUNDEDNESS_X1_X2} is satisfied, and the constant term of the lookahead distance is chosen accordingly to 
\begingroup
\setlength{\belowdisplayskip}{1pt}
\begin{equation}
\mu > \frac{4X_\text{max}}{Y_\text{min}-X_\text{max}\kappa_\text{max}}, \label{eq:nsb:analysis:barycenter:boundedness_X2:min_mu}
\end{equation}
\endgroup
where $X_\text{max} \triangleq |X(u, u_c)|_\infty$ and $\kappa_\text{max} \triangleq \max_{\theta \in P}|\kappa(\theta)|$.
\end{lemma}
\begin{proof}
\ifincludeproofs
The proof of this lemma is given in Appendix B.
\fi
\end{proof}
\begin{theorem}\label{th:nsb:analysis:barycenter:barycenter_usges}
Consider a $\theta$-parametrized path denoted by $P(\theta) = \left(x_p(\theta), y_p(\theta)\right)$, with the update law \eqref{eq:nsb:tasks:barycenter:path:update_law} and a system given by two vessels, each described by \eqref{eq:model:component_form}, giving the barycenter kinematics \eqref{eq:nsb:tasks:barycenter:kinematics:barycenter_kinematics_component_form}. Furthermore, let the adaptive controllers \eqref{eq:nsb:analysis:heading_controller} and \eqref{eq:nsb:analysis:surge_controller} be used as autopilots for each of the vessels, with the guidance law \eqref{eq:nsb:tasks:barycenter:guidance:guidance_law_heading}. Then, under the conditions of \crefrange{LEM:NSB:ANALYSIS:BARYCENTER:FORWARD_COMPLETENESS}{LEM:NSB:ANALYSIS:BARYCENTER:BOUNDEDNESS_X2}, the barycenter follows the path $P$ at the desired along-path speed $U_d(t)$ with bounded estimation errors and sway velocity, and the origin of the closed-loop system \eqref{eq:nsb:analysis:barycenter:closed_loop:X1}--\eqref{eq:nsb:analysis:barycenter:closed_loop:X2} is an USGES equilibrium point. 
\end{theorem}
\begin{proof}
    The proof follows along the lines of \cite[Proof of Theorem 1]{Belleter2019} for single vessel control but extended to two vessels and making use of the results in \cite{Pettersen2017} to prove USGES.

    First, consider the unactuated sway-dynamics \eqref{eq:nsb:analysis:barycenter:closed_loop:sway}. From \cite[Proposition 1]{Moe2016} the origin of \eqref{eq:nsb:analysis:barycenter:closed_loop:X2} is UGES. Moreover, by \cref{LEM:NSB:ANALYSIS:BARYCENTER:FORWARD_COMPLETENESS} the closed-loop system \eqref{eq:nsb:analysis:barycenter:closed_loop} is forward complete and thus the sway-dynamics \eqref{eq:nsb:analysis:barycenter:closed_loop:sway} is bounded near the manifold $\tilde{\mathbf{X}}_2 = 0$. Thus, we can conclude that there exists a finite time $T > t_0$ such that the solutions of \eqref{eq:nsb:analysis:barycenter:closed_loop:X2} will be sufficiently close to $\tilde{\mathbf{X}}_2 = 0$ to guarantee boundedness of $v_i$.
    
    Having established that the sway dynamics are bounded, we will now utilize cascaded theory to analyze the cascade \eqref{eq:nsb:analysis:barycenter:closed_loop:X1}--\eqref{eq:nsb:analysis:barycenter:closed_loop:X2}, where \eqref{eq:nsb:analysis:barycenter:closed_loop:X2} perturbs the nominal dynamics \eqref{eq:nsb:analysis:barycenter:closed_loop:X1} through the interconnection term $G_1(\cdot)$. Note that the estimation errors, and also the sway velocity which affects \eqref{eq:nsb:analysis:barycenter:closed_loop:X1}--\eqref{eq:nsb:analysis:barycenter:closed_loop:X2} through $U_d$, can be treated as time-varying signals in the following analysis since the the system is forward complete by \cref{LEM:NSB:ANALYSIS:BARYCENTER:FORWARD_COMPLETENESS}.

    First, consider the nominal dynamics given by the first term of \eqref{eq:nsb:analysis:barycenter:closed_loop:X1}. Taking the derivatives of the positive definite $\mathcal{C}^1$ Lyapunov function candidate
    \begingroup
\setlength{\belowdisplayskip}{1pt}
\setlength{\abovedisplayskip}{1pt}
    \begin{equation}
        V\left(\tilde{\mathbf{X}}_1\right) = \frac{1}{2}\left(x_{pb}^p\right)^2 + \frac{1}{2}\left(y_{pb}^p\right)^2,\label{eq:nsb:analysis:barycenter:barycenter_usges:lfc}
    \end{equation}
    \endgroup
    along the trajectories of \eqref{eq:nsb:analysis:barycenter:closed_loop:X1} gives
    \begin{equation}
        \dot{V} = -\tilde{\boldsymbol{X}}_1^T \mathbf{Q}\tilde{\boldsymbol{X}}_1 <0,\label{eq:nsb:tasks:barycenter:guidance:lyapunov_function_derivative}
    \end{equation}
    for which 
    \begingroup
    \setlength{\belowdisplayskip}{1pt}
    \begin{equation}
        \mathbf{Q} = \begin{bmatrix}
           \frac{k_\theta}{\sqrt{1 + \tilde{X}_{11}^2}} & 0  \\
           0 & \frac{1}{2}\frac{U_{d,1} + U_{d,2}}{\sqrt{\mu + \tilde{X}_{11}^2 +2 \tilde{X}_{21}^2}} 
           \end{bmatrix}>0
    \end{equation}
    \endgroup
    is a positive definite matrix as $k_\theta, U_{d,1}, U_{d,2}>0$,
    implying that $\dot{V}$ is negative definite, and that the nominal system is UGAS. Furthermore, to investigate USGES, the following bound can be verified to hold $\forall \tilde{\boldsymbol{X}}_1\in \mathcal{B}_r$
    \begin{equation}
        \dot{V} \leq -q_\text{min}\lVert\tilde{\boldsymbol{X}}_1\rVert^2,
    \end{equation}
    with
    \begin{equation}
        q_\text{min} \triangleq \lambda_\text{min}\left( \begin{bmatrix}
        \frac{k_\theta}{\sqrt{1 + r^2}} & 0  \\
        0 & \frac{1}{2}\frac{U_{d,1} + U_{d,2}}{\sqrt{\mu + 3 r^2}} 
        \end{bmatrix}\right)
    \end{equation}
    for any ball $\mathcal{B}_r \triangleq \left\{\max\{|\tilde{X}_{11}|, |\tilde{X}_{21}|\} < r\right\}$, $r>0$, where $\lambda_\text{min}(\mathbf{A})$ is defined as the minimum eigenvalue of $\mathbf{A}$. Thus, the conditions of \cite[Theorem 5]{Pettersen2017} is fulfilled with $k_1 = k_2 = \frac{1}{2}$, $a=2$ and $k_3 = q_\text{min}$, and USGES can be concluded for the origin of the nominal system given by the first term of \eqref{eq:nsb:analysis:barycenter:closed_loop:X1}.

    The perturbing system \eqref{eq:nsb:analysis:barycenter:closed_loop:X2} is proven UGES in \cite{Moe2016}, implying both UGAS and USGES. The conditions of \cite[Theorem 5]{Pettersen2017} are therefore trivially satisfied for the perturbing system.
    
    The existence of positive constants $c_1, c_2, \eta>0$ satisfying \cite[Assumption 1]{Pettersen2017} is clearly satisfied by $V$ in \eqref{eq:nsb:analysis:barycenter:barycenter_usges:lfc}:
    \begingroup
\setlength{\belowdisplayskip}{1pt}
    \begin{align}
        \left\lVert\frac{\partial V}{\partial \tilde{\mathbf{X}}_1}\right\rVert\!\left\lVert\tilde{\mathbf{X}}_1\right\rVert\!&=\!\left\lVert\left[x_{pb}^p, y_{pb}^p\right]^T\right\rVert^2\!=\!2V\left(\tilde{\mathbf{X}}_1\right)\;\forall \left\lVert\tilde{\mathbf{X}}_1\right\rVert\\
        \left\lVert \frac{\partial V}{\partial \tilde{\mathbf{X}}_1}\right\rVert &= \left\lVert\tilde{\mathbf{X}}_1\right\rVert \leq \eta \quad \forall \left\lVert\tilde{\mathbf{X}}_1\right\rVert \leq \eta,
    \end{align}
    \endgroup
    i.e. with $c_1 = 2$ and $c_2 = \eta$ for any choice $\eta > 0$. Finally, the conditions of \cite[Assumption 2]{Pettersen2017} must be investigated, i.e. the assumption that the interconnection terms, the second vector of \eqref{eq:nsb:analysis:barycenter:closed_loop:X1}, has at most linear growth in $\tilde{\mathbf{X}}_1$. From \eqref{eq:nsb:tasks:barycenter:guidance:G_1_features:linear} it can be seen that the interconnection term does not grow with the states $\tilde{\mathbf{X}}_1$ as it can be bounded by linear functions of $\tilde{\mathbf{X}}_2$. All conditions of \cite[Proposition 9]{Pettersen2017} are therefore satisfied, and the origin of the closed-loop system \eqref{eq:nsb:analysis:barycenter:closed_loop:X1}--\eqref{eq:nsb:analysis:barycenter:closed_loop:X2} $\left(\tilde{\mathbf{X}}_1, \tilde{\mathbf{X}}_2\right) = (\mathbf{0}, \mathbf{0})$, is USGES and UGAS. 
    
    Boundedness of \crefrange{eq:nsb:analysis:barycenter:closed_loop:theta_r}{eq:nsb:analysis:barycenter:closed_loop:theta_u} is established in the proof of \cite[Proposition 1]{Moe2016} where it is shown that both $\tilde{\boldsymbol{\theta}}_r$ and $\tilde{\boldsymbol{\theta}}_r$ are bounded. This concludes the proof of \cref{th:nsb:analysis:barycenter:barycenter_usges}.
\end{proof}

\section{SIMULATIONS}\label{sec:sim}
In this section, we present the results from numerical simulations of two identical underactuated USVs modeled by \eqref{eq:model:component_form}, subject to a constant irrotational ocean current $\mathbf{V}_c\triangleq\left[-0.707, -0.707, 0\right]^T$. The NSB collision avoidance and vessel formation task objectives \eqref{eq:control_objectives:collision_avoidance_objective}--\eqref{eq:control_objectives:formation_objective} are specified by $\sigma_{ca, d} = \SI{20}{\metre}$ and $\boldsymbol{\sigma}_{f,d} = [0, \SI{20}{\metre}]^T$, with the control gains $\lambda_{ca} = 1$ and $\boldsymbol{\Lambda}_f^p = \diag(2.5, 0.3)$, respectively. The desired surge speed is chosen constant as $u_d = \SI{3}{\metre\per\second}$ and the desired path to follow is defined as 
\begin{equation}
    P \triangleq \begin{cases}
    x_p(\theta) = \theta\\
    y_p(\theta) = 300\sin\left(0.005\theta\right),
    \end{cases}
\end{equation}
which satisfies the condition of \Cref{LEM:NSB:ANALYSIS:BARYCENTER:BOUNDEDNESS_X1_X2} as $\max_{\theta \in P}|\kappa(\theta)| = 0.0075 < Y_\text{min}/X_\text{max} \approx 0.0882$. Furthermore, the required bound for $\mu$ to satisfy the condition of \Cref{LEM:NSB:ANALYSIS:BARYCENTER:BOUNDEDNESS_X2} can be calculated to be $\mu > \SI{49.5704}{\metre}$, which is satisfied by choosing $\mu = \SI{50}{\metre}$. The controller gains are chosen as $k_\psi = 1.2$, $k_r = 1.3$, $\lambda = 100$, $k_d = 10$, $k_u = 0.1$ and $k_e = 0.1$, with the adaptive gains are chosen as $\gamma_r = 5$ and $\gamma_u = 1$.
\begin{figure}[htbp]
    \centering
    \includegraphics[width=\columnwidth]{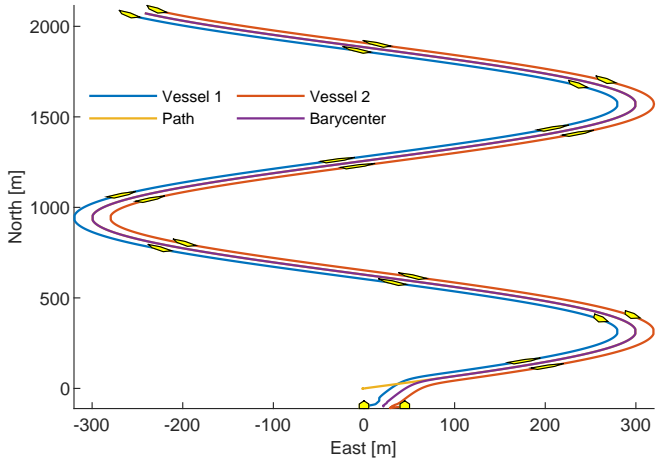}
    \vspace*{-7mm}
    \caption{Path following of the desired sinusiodal path.}
    \label{fig:sim:ideal:all_tasks:ned}
\end{figure}

The resulting trajectories of both vessels and the barycenter trajectory are shown in \cref{fig:sim:ideal:all_tasks:ned}. It can be observed that the vessels maintain their desired formation while making the barycenter follow the desired path. It can also be seen that the curved path and the ocean currents make both vessels operate with a non-zero sideslip angle. The path following errors of the three tasks can be seen in \cref{fig:sim:ideal:all_tasks:error} which shows the barycenter task errors converging to zero. Moreover, it can be observed that the vessel formation task errors grow during the turns while converging towards zero elsewhere.
\begin{figure}[htbp]
    \centering
    \includegraphics[width=\columnwidth]{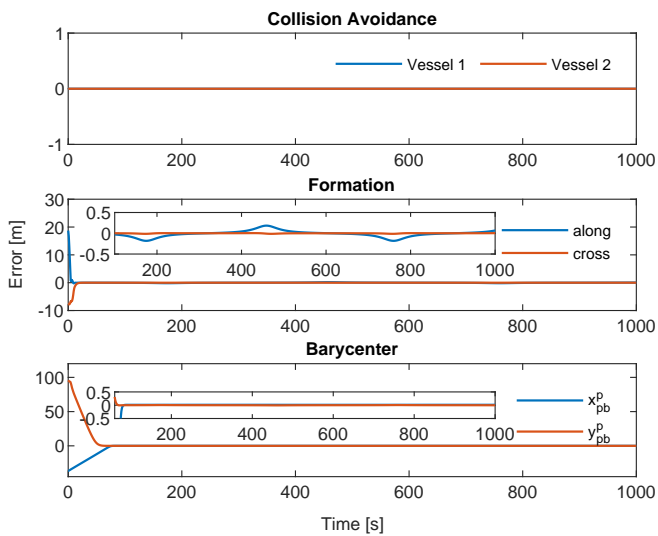}
    \vspace*{-7mm}
    \caption{NSB errors of the desired sinusiodal path.}
    \label{fig:sim:ideal:all_tasks:error}
\end{figure}

\vspace*{-3mm}
\section{EXPERIMENTS}
In this section, results from experiments at sea are presented. The experiments were performed using the Odin and Frigg USVs which are under development by FFI. They are $\SI{11}{\metre}$ long and $\SI{3.5}{\metre}$ wide and propelled by a dual waterjet system. However, at maneuvering speeds, the waterjets are linked together, rendering the system underactuated. 

We did not have the opportunity to implement new autopilots at Odin and Frigg, so the adaptive autopilots \eqref{eq:nsb:analysis:heading_controller}--\eqref{eq:nsb:analysis:surge_controller} could not be implemented. Instead, the existing surge and heading autopilots had to be used. These are PI and PD controllers for surge and heading, respectively, which are tuned to give asymptotic stability. 

The desired along-path speed was chosen constant as $u_d = \SI{3}{\metre\per\second}$, while the lookahead distance was chosen as $\mu = \SI{100}{\metre}$. The NSB task objectives of the collision avoidance and vessel formation tasks were chosen according to \cref{sec:sim}, with the NSB task gains chosen as $\lambda_{ca} = 1$ and $\boldsymbol{\Lambda}_f^p = \diag(0.3, 0.1)$, respectively.
\begin{figure}[htbp]
    \centering
    \includegraphics[width=\columnwidth]{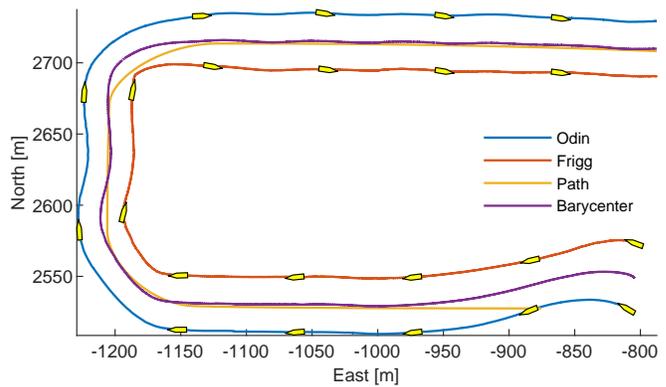}
    \vspace*{-7mm}
    \caption{The vessels positions along with desired and actual barycenter path. }
    \label{fig:exp:case:1:ned}
\end{figure}

The resulting motion of the vessels is shown in \cref{fig:exp:case:1:ned}. The vessels maintained the desired formation on the straight-line path segments, while the error was bounded during turns, as shown in \cref{fig:exp:case:1:error_odin}. For the barycenter task, a small cross-track error of $1 - \SI{2}{\metre}$ can be observed. We believe that the reason for this deviation is the lack of ocean current adaptation in the existing autopilots, and that the adaptation of \eqref{eq:nsb:analysis:heading_controller}--\eqref{eq:nsb:analysis:surge_controller} could significantly reduce the cross-track error as predicted by theory and illustrated in the simulations. Moreover, rather large vessel formation task errors can be observed, especially during the second turn. We believe that these errors are caused by sub-optimal NSB task gains caused by the limited timespan in which the experiments had to be executed, and that these errors could be significantly reduced in a well-tuned system. 
\begin{figure}[htbp]
    \centering
    \includegraphics[width=\columnwidth]{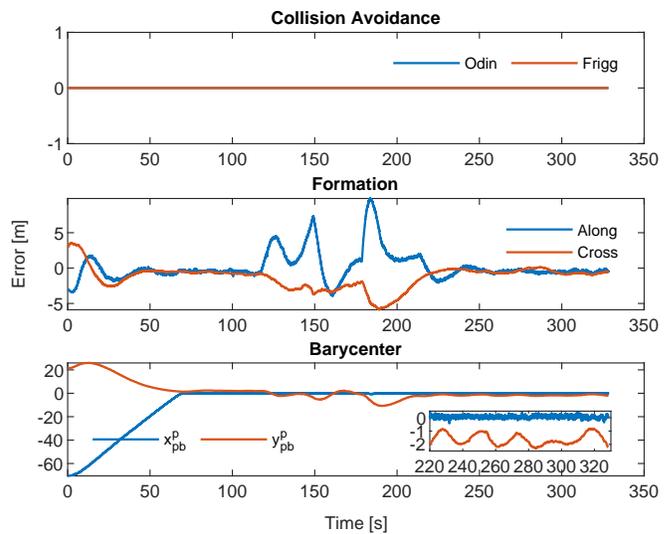}
    \vspace*{-7mm}
    \caption{NSB task errors.}
    \label{fig:exp:case:1:error_odin}
\end{figure}

\vspace*{-3mm}
\section{CONCLUSION}
In this paper we have proposed a formation control method for two underactuated USVs to follow curved paths in the presence of ocean currents. This is achieved by integrating a LOS guidance law for curved path following of the barycenter into the NSB framework together with the additional two tasks: collision avoidance and vessel formation. We prove that the proposed LOS guidance law combined with adaptive feedback linearizing controllers with sliding mode for the surge and yaw autopilots, achieves convergence to the desired path, and that the closed-loop system is USGES and UGAS while the underactuated sway dynamics remains bounded. The resulting control system requires only traditional sensors for estimating absolute velocities, such as IMU and GNSS, removing the need of expensive sensors for measuring relative velocities. Both simulation and experimental results are presented to validate the theoretical results. 
\endgroup




\vspace*{-2mm}
\useRomanappendicesfalse
\appendices%
\section{}\label{app:model_expressions}%
\begingroup
\setlength{\belowdisplayskip}{0pt}
\setlength{\abovedisplayskip}{0pt}
\vspace*{-3mm}
\begin{equation}
        \boldsymbol{\phi}_u^T(\psi, r)\!=\!\begin{bmatrix}
        - \frac{d_{11} + 2d_{11}^q u}{m_{11}}\cos(\psi) - \frac{m_{11}^A - m_{22}^A}{m_{11}}r\sin(\psi)\\
        - \frac{d_{11} + 2d_{11}^q u}{m_{11}}\sin(\psi) + \frac{m_{11}^A - m_{22}^A}{m_{11}}r\cos(\psi)\\
        -d_{11}^q\cos^2(\psi)\\
        -d_{11}^q\sin^2(\psi)\\
        -2d_{11}^q\cos(\psi)\sin(\psi)
        \end{bmatrix}\!\!
\end{equation}
\begin{align}
&X(u, u_c) = \frac{1}{\Gamma}\big(m_{33}\left(-d_{23}-m_{11}(u - u_c) - m_{11}^{RB}u_c\right)\nonumber\\
+ &m_{23}d_{33} + m_{23}\left(m_{23}(u - u_c)+m_{23}^{RB}u_c + m_{22}^Au_c\right)\big)\!\!\label{eq:appendix:vessel_model_expressions:X}
\end{align}
\begin{align}
Y(u, u_c) &= \frac{1}{\Gamma}\big(-m_{33}d_{22} + m_{23}d_{32} \nonumber\\
&\hphantom{{}=} + m_{23}\left(m_{22}^A - m_{11}^A\right) (u - u_c)\big)\label{eq:appendix:vessel_model_expressions:Y}
\end{align}
\begin{align}
&F_r(u, v, r) = - \frac{m_{23}}{\Gamma}\left(-m_{11}ru - d_{22}v - d_{23}r\right) +\frac{m_{22}}{\Gamma}\big( \nonumber\\
&\hphantom{{}=} -\big(m_{22}v- m_{23}r\big)u + m_{11}uv - d_{32}v - d_{33}r\big)
\end{align}
where, $\Gamma = m_{22}m_{33}-m_{23}^2>0$. Further, the function $\boldsymbol{\phi}_r^T(u, v, r, \psi) = \left[\phi_{r1}, \dots, \phi_{r5}\right]$ is given by
\begingroup
\allowdisplaybreaks
\begin{align}
\begin{bmatrix}\phi_{r1}\\\phi_{r2}\end{bmatrix} &= \begin{bmatrix}
\cos(\psi) & -\sin(\psi)\\
\sin(\psi) & \cos(\psi)
\end{bmatrix} \begin{bmatrix}a_1\\a_2\end{bmatrix}\\
\phi_{r3} &= -\frac{m_{22}}{\Gamma}\left(m_{11}^A - m_{22}^A\right)\cos(\psi)\sin(\psi)\\
\phi_{r4} &= \frac{m_{22}}{\Gamma}\left(m_{11}^A - m_{22}^A\right)\cos(\psi)\sin(\psi)\\
\phi_{r5} &=  \frac{m_{22}}{\Gamma}\left(m_{11}^A - m_{22}^A\right)\left(1 - 2\sin^2(\psi)\right)
\end{align}
\endgroup
where
\begingroup
\allowdisplaybreaks
\begin{align}
a_1 &= \frac{m_{22}}{\Gamma}\!\left(\!\left(m_{11}^A\!-\! m_{22}^A\!\right)v\!+\!\!\left(m_{23}^A\!-\! m_{22}^A\!\right)r\!\right)\!-\!\frac{m_{23}}{\Gamma}m_{11}^Ar\!\!\\
a_2 &= \frac{m_{22}}{\Gamma}\left(d_{32} - \left(m_{11}^A-m_{22}^A\right)u\right) - \frac{m_{23}}{\Gamma}d_{22}.
\end{align}
\endgroup

\section*{ACKNOWLEDGMENT}
The authors would like to thank Jarle Sandrib, Geir Lofsberg, and Fredrik Hermansen at FFI for their valuable help and assistance during the experiments.
\endgroup

\ifincludeproofs
\onecolumn
\section{}
\subsection{Proof of Lemma 1}
\allowdisplaybreaks
The proof follows along the lines of \cite[Lemma 1]{Belleter2018Proofs} but is extended to two vessels described by \eqref{eq:model:component_form} expressed in terms of absolute velocities, with the adaptive controllers \eqref{eq:nsb:analysis:heading_controller} - \eqref{eq:nsb:analysis:surge_controller}.

First, consider forward completeness of the underactuated sway dynamics \eqref{eq:nsb:analysis:barycenter:closed_loop:sway}. From the boundedness of the vector $[\tilde{\mathbf{X}}_{2,i}^T, \kappa(\theta), u_{d,i}, \dot{u}_{d,i}, u_c, v_c]^T$ there exist some scalar $\beta_0\in\mathbb{R}_{>0}$ such that $\lVert[\tilde{\mathbf{X}}_{2,i, \kappa(\theta), u_{d,i}, \dot{u}_{d,i}, u_c, v_c}^T]^T\rVert\leq \beta_0$. Moreover, from \eqref{eq:nsb:analysis:barycenter:r_d} we can concluded the existence of some positive functions $a_{r_d}(\cdot)$ and $b_{r_d}(\cdot)$ such that 
\begin{equation}
    |r_d(\cdot)| \leq a_{r_d}(\mu, \beta_0) |v| + b_{r_d}(\mu, \beta_0).
\end{equation}
Then, choosing the Lyapunov function candidate (LFC), omitting subscripts for simplicity
\begin{equation}
    V_1(v) = \frac{1}{2}v^2,
\end{equation}
whose time derivative along the solutions of \eqref{eq:nsb:analysis:barycenter:closed_loop:sway} is:
\begin{align}
    \dot{V}_1(v) &= X(u_d + \tilde{u}, u_c)r_dv + X(u_d + \tilde{u}, u_c)\tilde{r} v\nonumber\\
        &\hphantom{{}=}+ Y(u_d + \tilde{u}, u_c)v^2 - Y(u_d + \tilde{u}, u_c)v_cv
\end{align}
Using Young's inequality, we conclude that the following bound holds:
\begin{align}
    \dot{V}_1(v) &\leq  Y(u_d + \tilde{u}, u_c)v^2 + X(u_d + \tilde{u}, u_c)(\tilde{r}^2+ v^2)\nonumber\\
        &\hphantom{{}=}+ X(u_d + \tilde{u}, u_c)(r_d^2 + v^2) - Y(u_d + \tilde{u}, u_c)(v_c^2 + v^2)\\
        &\leq \alpha V + \beta\label{eq:nsb:analysis:barycenter:forward_completeness:sway_lyapunov_bounded}
\end{align}
where $\alpha\in\mathbb{R}_{\geq0}$, $\beta\in\mathbb{R}_{\geq0}$ are positive scalars. As \eqref{eq:nsb:analysis:barycenter:forward_completeness:sway_lyapunov_bounded} is a scalar system, the comparison lemma \cite[Lemma 3.4]{Khalil2002} may be used to bound the solutions of \eqref{eq:nsb:analysis:barycenter:forward_completeness:sway_lyapunov_bounded} by the scalar linear system
\begin{equation}
    \dot{x} = \alpha x + \beta
\end{equation}
whose solution is equal to 
\begin{equation}
    x(t) = \frac{\lVert x(t_0)\rVert\alpha + \beta}{\alpha}e^{\alpha(t-t_0)} - \frac{\beta}{\alpha}.
\end{equation}
Hence, by the comparison lemma, the solutions of \eqref{eq:nsb:analysis:barycenter:forward_completeness:sway_lyapunov_bounded} must be upper bounded by
\begin{equation}
    V_1(v) \leq\frac{\lVert x(t_0)\rVert\alpha + \beta}{\alpha}e^{\alpha(t-t_0)} - \frac{\beta}{\alpha}.
\end{equation}
As $V(v)$ is defined for all $t$ up to $t_\text{max} = \infty$, it follows that $v$ must also be defined up to $t_\text{max} = \infty$. In the same way as \cite{Belleter2018Proofs}, the solutions of \eqref{eq:nsb:analysis:barycenter:closed_loop:sway} thus fulfills the definition of forward completeness in \cite{Angeli1999} and forward completeness of the solution of \eqref{eq:nsb:analysis:barycenter:closed_loop:sway} can be concluded.

Forward completeness of the closed-loop system \eqref{eq:nsb:analysis:barycenter:closed_loop:X2} - \eqref{eq:nsb:analysis:barycenter:closed_loop:theta_u} is established in \cite[Proposition 1]{Moe2016}. Having established forward completeness of \eqref{eq:nsb:analysis:barycenter:closed_loop:X2} - \eqref{eq:nsb:analysis:barycenter:closed_loop:sway}, only the forward completeness of \eqref{eq:nsb:analysis:barycenter:closed_loop:X1} remains before forward completeness may be concluded for the whole closed-loop system \eqref{eq:nsb:analysis:barycenter:closed_loop}. To show forward completeness of the along- and cross-track error dynamics, consider the LFC
\begin{equation}
    V_2 = \frac{1}{2}\left(x_{pb}^p\right)^2 + \frac{1}{2}\left(y_{pb}^p\right)^2,
\end{equation}
whose derivative along the solutions of \eqref{eq:nsb:analysis:barycenter:closed_loop:X1} is 
\begin{align}
    \dot{V}_2 &= - k_\theta \frac{\left(x_{pb}^p\right)^2}{\sqrt{1 + \left(x_{pb}^p\right)^2}} -\frac{1}{2}\left(U_{d,1} + U_{d,2}\right)\frac{\left(y_{pb}^p\right)^2}{\sqrt{\Delta^2 + \left(y_{pb}^p\right)^2}} + G_1(\cdot)y_{pb}^p\\
    &\leq G_1(\cdot)y_{pb}^p + \left(x_{pb}^p\right)^2
\end{align}
Using Young's inequality, along with the \eqref{eq:nsb:tasks:barycenter:guidance:G_1_features} we obtain the following bound 
\begin{align}
    \dot{V}_2 &\leq V_2 + \frac{1}{2}\zeta_1^2\left(U_{d,1}, U_{d,2}\right) \left\lVert[\tilde{\psi}_1, \tilde{u}_1, \tilde{\psi}_2, \tilde{u}_2]^T\right\rVert^2\label{eq:nsb:analysis:barycenter:forward_completeness:V2_bounded}\\
    &\leq V_2 + \sigma_2\left(v_1, v_2, \tilde{\psi}_1, \tilde{u}_1, \tilde{\psi}_2, \tilde{u}_2\right),
\end{align}
where $\sigma_2(\cdot)\in\mathcal{K}_\infty$. By viewing the arguments of $\sigma_2(\cdot)$ as inputs to the along- and cross-track error dynamics, \cite[Corollary 2.11]{Angeli1999} is satisfied by \eqref{eq:nsb:analysis:barycenter:forward_completeness:V2_bounded} and forward completeness of the solutions of \eqref{eq:nsb:analysis:barycenter:closed_loop:X1} can be concluded. Similarly to \cite{Belleter2018Proofs}, the arguments of $\sigma_2(\cdot)$ are all forward complete, and are therefore valid input signals according to \cite{Angeli1999}. Forward completeness for the whole closed-loop system \eqref{eq:nsb:analysis:barycenter:closed_loop} is therefore established, concluding the proof of \cref{LEM:NSB:ANALYSIS:BARYCENTER:FORWARD_COMPLETENESS}.

\subsection{Proof of Lemma 2}
This proof follows along the lines of \cite[Lemma 2]{Belleter2018Proofs} but is extended to two vessels described by \eqref{eq:model:component_form} expressed in terms of absolute velocities, with the adaptive controllers \eqref{eq:nsb:analysis:heading_controller} - \eqref{eq:nsb:analysis:surge_controller}.

To prove boundedness of $v$ near the manifold $(\tilde{\mathbf{X}}_1, \tilde{\mathbf{X}}_2) = \mathbf{0}$, recall the sway dynamics \eqref{eq:nsb:analysis:barycenter:closed_loop:sway}:
\begin{align}
    \dot{v}_i &= X(u_{d,i} + \tilde{u}_i, u_c)r_{d_i} + X(u_{d,i} + \tilde{u}_i, u_c)\tilde{r}_i \nonumber\\
    &\hphantom{{}=}+ Y(u_{d,i} + \tilde{u}_i, u_c)v_i - Y(u_{d,i} + \tilde{u}_i, u_c)v_c.
\end{align}
We then consider the Lyapunov function candidate $V(v_i) = \frac{1}{2}v_i^2$, whose time derivative along the solutions of \eqref{eq:nsb:analysis:barycenter:closed_loop:sway} is:
\begin{align}
    \dot{V} &= v_i\dot{v}_i =  X(u_{d,i} + \tilde{u}_i, u_c)r_{d_i}v_i + X(u_{d,i} + \tilde{u}_i, u_c)\tilde{r}_i v_i \nonumber\\
    &\hphantom{{}=}+ Y(u_{d,i} + \tilde{u}_i, u_c)v_i^2 - Y(u_{d,i} + \tilde{u}_i, u_c)v_c v_i\\
    &\leq X(u_{d,i}, u_c)r_{d_i} v_i + a_x\tilde{u}_i r_{d,i} v_i + X(u_{d,i}, u_c)\tilde{r}_i v_i+ a_x\tilde{u}_i \tilde{r}_i v_i\nonumber\\
    &\hphantom{{}=}+ Y(u_{d,i}, u_c)v_i^2 + a_y \tilde{u}_i v_i^2 - Y(u_{d,i}, u_c)v_c v_i - a_y \tilde{u}_i v_c v_i.\label{eq:appendix:proofs:boundedness_x1_x2:v_dot_before_inserting_rd_v}
\end{align}
Here, we have used the following properties of  $X(u, u_c)$ and $Y(u, u_c)$ from \eqref{eq:appendix:vessel_model_expressions:X} - \eqref{eq:appendix:vessel_model_expressions:Y}:
\begin{align}
    X(u, u_c) &= a_x u + b_x u_c + c_x \\
    Y(u, u_c) &= a_y u + b_y u_c + c_y.
\end{align}
Next, to find an upper bound of the term $r_{d,i} v_i$ in \eqref{eq:appendix:proofs:boundedness_x1_x2:v_dot_before_inserting_rd_v}, we substitute the expression for $r_d$ from \eqref{eq:nsb:analysis:barycenter:r_d}, omitting subscripts for simplicity:
\begin{align}
    r_d v &= \kappa(\theta) \dot{\theta} v + \frac{\dot{u}_d }{u_d^2 + v^2} v^2 - \frac{u_d \dot{v}v}{u_d^2 + v^2}\nonumber\\
    &\hphantom{{}=}- \frac{v}{\Delta^2 + \left(y_{pb}^p\right)^2}\left[\Delta \dot{y}_{pb}^p - y_{pb}^p\left(\frac{\partial \Delta}{\partial x_{pb}^p}\dot{x}_{pb}^p + \frac{\partial \Delta}{\partial y_{pb}^p}\dot{y}_{pb}^p\right)\right] \\
    &= \kappa(\theta) v \left(\frac{1}{2}U_1\cos\left(\chi_1 - \gamma_p\right) + \frac{1}{2}U_2\cos\left(\chi_2 - \gamma_p\right) + \frac{k_\theta x_{pb}^p}{\sqrt{1 + \left(x_{pb}^p\right)^2}} \right) \nonumber\\
    &\hphantom{{}=} + \frac{\dot{u}_d }{u_d^2 + v^2} v^2 - \frac{u_dv}{u_d^2 + v^2}\Big(X(u, u_c) r + Y(u,u_c)v - Y(u,u_c)v_c\Big)\nonumber\\
    &\hphantom{{}=} -\frac{\Delta v}{\Delta^2 + \left(y_{pb}^p\right)^2}\left(- \frac{1}{2}\left(U_{d,1} + U_{d,2}\right)\frac{y_{pb}^p}{\sqrt{\Delta^2 + \left(y_{pb}^p\right)^2}} - \kappa(\theta)\dot{\theta}x_{pb}^p + G_1(\cdot)\right)\nonumber\\
    &\hphantom{{}=} +\frac{y_{pb}^p v}{\Delta^2 + \left(y_{pb}^p\right)^2}\bigggg[ \frac{\partial \Delta}{\partial x_{pb}^p}\left(- \frac{k_\theta x_{pb}^p}{\sqrt{1 + \left(x_{pb}^p\right)^2}} + \dot{\theta}\kappa(\theta)y_{pb}^p\right)\nonumber\\
    &\hphantom{{}=} + \frac{\partial \Delta}{\partial y_{pb}^p}\left(- \frac{1}{2}\left(U_{d,1} + U_{d,2}\right)\frac{y_{pb}^p}{\sqrt{\Delta^2 + \left(y_{pb}^p\right)^2}} - \kappa(\theta)\dot{\theta}x_{pb}^p + G_1(\cdot)\right)\bigggg].\label{eq:appendix:proofs:boundedness_x1_x2:r_dv_before_extracting_F}
\end{align}
Now, we introduce a term $F(\tilde{\mathbf{X}}_1, \tilde{\mathbf{X}}_2, \Delta, \theta, u_d, \dot{u}_d, v, v_c, u_c, r)$ to collect all terms that grows linearly with $v$ and the terms that grow quadratically with $v$ but 
vanish when $\tilde{\mathbf{X}}_1$ and $\tilde{\mathbf{X}}_2$ are zero:
\begin{align}
     r_d v &= v \left[1 + \frac{\Delta x_{pb}^p}{\Delta^2 + \left(y_{pb}^p\right)^2}\right] \kappa(\theta) \left(\frac{1}{2}U_1\cos\left(\chi_1 - \gamma_p\right) + \frac{1}{2}U_2\cos\left(\chi_2 - \gamma_p\right)\right)\nonumber\\
     &\hphantom{{}=} -\frac{u_d}{u_d^2 + v^2} Y(u, u_c)v^2 + F(\tilde{\mathbf{X}}_1, \tilde{\mathbf{X}}_2, \Delta, \theta, u_d, \dot{u}_d, v, v_c, u_c, r),\label{eq:appendix:proofs:boundedness_x1_x2:r_dv_after_extracting_F}
\end{align}
where the expression for $\dot{\theta}$ has been inserted in the second last term on the third line in \eqref{eq:appendix:proofs:boundedness_x1_x2:r_dv_before_extracting_F} to extract the second term in the first parenthesis on the first line of \eqref{eq:appendix:proofs:boundedness_x1_x2:r_dv_after_extracting_F} and $F(\cdot)$ is given by
\begin{align}
    F(\cdot) &= v \bigggg[\kappa(\theta) \frac{k_\theta x_{pb}^p}{\sqrt{1 + \left(x_{pb}^p\right)^2}} + \frac{\dot{u}_d }{u_d^2 + v^2} v^2 - \frac{u_dv}{u_d^2 + v^2}\Big(X(u, u_c) r - Y(u,u_c)v_c\Big)\nonumber\\ 
    &\hphantom{{}=}-\frac{\Delta}{\Delta^2 + \left(y_{pb}^p\right)^2}\left(- \frac{1}{2}\left(U_{d,1} + U_{d,2}\right)\frac{y_{pb}^p}{\sqrt{\Delta^2 + \left(y_{pb}^p\right)^2}} - \kappa(\theta)\frac{k_\theta \left(x_{pb}^p\right)^2}{\sqrt{1 + \left(x_{pb}^p\right)^2}} + G_1(\cdot)\right)\nonumber\\
    &\hphantom{{}=}-\frac{y_{pb}^p v}{\Delta^2 + \left(y_{pb}^p\right)^2}\left[\frac{\partial \Delta}{\partial x_{pb}^p}\frac{k_\theta x_{pb}^p}{\sqrt{1 + \left(x_{pb}^p\right)^2}} + \frac{\partial \Delta}{\partial y_{pb}^p}\left(- \frac{1}{2}\left(U_{d,1} + U_{d,2}\right)\frac{y_{pb}^p}{\sqrt{\Delta^2 + \left(y_{pb}^p\right)^2}} + G_1(\cdot)\right)\right]\bigggg].\label{eq:appendix:proofs:boundedness_x1_x2:F}
\end{align}
Note how all terms with partial derivatives of $\Delta$ and $\dot{\theta}$ are cancelled due to skew-symmetry from the definition of the lookahead distance \eqref{EQ:NSB:TASKS:BARYCENTER:GUIDANCE:LOOKAHEAD}:
\begin{align}
    &\frac{\Delta}{x_{pb}^p} \dot{\theta}\kappa(\theta)y_{pb}^p - \frac{\Delta}{y_{pb}^p} \dot{\theta}\kappa(\theta)x_{pb}^p= \frac{x_{pb}^p}{\sqrt{1 + \left(x_{pb}^p\right)^2 + \left(y_{pb}^p\right)^2}}\dot{\theta}\kappa(\theta)y_{pb}^p - \frac{y_{pb}^p}{\sqrt{1 + \left(x_{pb}^p\right)^2 + \left(y_{pb}^p\right)^2}}\dot{\theta}\kappa(\theta)x_{pb}^p = 0.
\end{align}
Furthermore, from \eqref{eq:appendix:proofs:boundedness_x1_x2:F} it can be seen that the function $F(\cdot)$ may be upper bounded by the following inequality
\begin{equation}
    |F(\cdot)| \leq F_2(\tilde{\mathbf{X}}_1, \tilde{\mathbf{X}}_2, \Delta, \theta, u_d, \dot{u}_d, v, v_c, u_c, r)v^2 + F_1(\tilde{\mathbf{X}}_1, \tilde{\mathbf{X}}_2, \Delta, \theta, u_d, \dot{u}_d, v, v_c, u_c, r) |v|,
\end{equation}
where $F_{1,2}(\cdot)$ are positive functions with 
\begin{equation}
    F_2(\mathbf{0}, \mathbf{0}, \Delta, \theta, u_d, \dot{u}_d, v, v_c, u_v) = 0.
\end{equation}
Consequently, the term $r_d v$ may be upper bounded as:
\begin{align}
    r_{d,i} v_i &\leq |v_i| \left|\left[ 1 + \frac{\Delta x_{pb}^p}{\Delta^2 + \left(y_{pb}^p\right)^2}\right]\right| |\kappa(\theta)| \frac{1}{2}\left(|U_i| + |U_j|\right) + |F(\cdot)| -\frac{u_{d,i}}{u_{d,i}^2 + v_i^2} Y(u_i, u_c)v_i^2\\
    &\leq |v_i| \left|\left[ 1 + \frac{\Delta x_{pb}^p}{\Delta^2 + \left(y_{pb}^p\right)^2}\right]\right| |\kappa(\theta)| \frac{1}{2}\left(|u_i| + |v_i| + |u_j|+ |v_j| \right) + |F(\cdot)| -\frac{u_{d,i}}{u_{d,i}^2 + v_i^2} Y(u_i, u_c)v_i^2\\
    &\leq \frac{1}{2}\left|\left[ 1 + \frac{\Delta x_{pb}^p}{\Delta^2 + \left(y_{pb}^p\right)^2}\right]\right||\kappa(\theta)| v_i^2 -\frac{u_{d,i}}{u_{d,i}^2 + v_i^2} Y(u_i, u_c)v_i^2 + |F(\cdot)| \nonumber\\
    &\hphantom{{}=}+\frac{1}{2}\left|\left[ 1 + \frac{\Delta x_{pb}^p}{\Delta^2 + \left(y_{pb}^p\right)^2}\right]\right||\kappa(\theta)|\left(|u_i|  + |u_j|+ |v_j| \right) |v_i|\\
    &\leq |\kappa(\theta)| v_i^2 -\frac{u_{d,i}}{u_{d,i}^2 + v_i^2} Y(u_i, u_c)v_i^2 + |F(\cdot)| +|\kappa(\theta)|\left(|u_i| + |u_j|+ |v_j| \right) |v_i|,\label{eq:appendix:proofs:boundedness_x1_x2:bound_rdv}
\end{align}
where we have used the following boundedness property:
\begin{align}
    \left|\left[ 1 + \frac{\Delta x_{pb}^p}{\Delta^2 + \left(y_{pb}^p\right)^2}\right]\right| & \leq 2. \label{eq:appendix:proofs:boundedness_x1_x2:1_plus_delta_x_less_than_two}
\end{align}

\begin{remark}
To be able to have $|F(\cdot)|$ be upper bounded by a quadratic function of $v$, it is necessary to choose $\Delta$ dependent on $x_{pb}^p$ in \eqref{EQ:NSB:TASKS:BARYCENTER:GUIDANCE:LOOKAHEAD}, but it is not required to be dependent on $y_{pb}^p$ for the conditions of \cref{LEM:NSB:ANALYSIS:BARYCENTER:BOUNDEDNESS_X1_X2} to hold. Contrary to \cite{Belleter2019}, where the lookahead distance had to depend on both $x_{pb}^p$ and $y_{pb}^p$. This, comes from our choice of defining the LOS guidance law in terms of absolute velocities, contrary to \cite{Belleter2019} where relative velocities are used. Thus, the LOS guidance law \eqref{eq:nsb:tasks:barycenter:guidance:guidance_law_heading} does not include the ocean current observer and the ocean current dependent term $g$ from \cite{Belleter2019}, as the ocean current compensation is instead handled by the adaptive autopilots \eqref{eq:nsb:analysis:heading_controller}-\eqref{eq:nsb:analysis:surge_controller}. Consequently, the term in \cite{Belleter2018Proofs} that could grow unbounded in $y_{b/p}$ near the manifold where $g = -(y_{b/p} + 1)$ is not present in \eqref{eq:appendix:proofs:boundedness_x1_x2:r_dv_before_extracting_F}, removing the dependency on $y_{pb}^p$ in \eqref{EQ:NSB:TASKS:BARYCENTER:GUIDANCE:LOOKAHEAD}. The proof can be found in Appendix B.4.
\end{remark}

Having established an upper bound of the term $r_d v$, we substitute \eqref{eq:appendix:proofs:boundedness_x1_x2:bound_rdv} into \eqref{eq:appendix:proofs:boundedness_x1_x2:v_dot_before_inserting_rd_v} to obtain the following bound for $\dot{V}$:
\begin{align}
    \dot{V} & \leq \left[|\kappa(\theta)|X(u_{d,i}, u_c) + Y(u_i, u_c)\right]v_i^2 + a_x\tilde{u}_i r_{d,i} v_i + X(u_{d,i}, u_c)\tilde{r}_i v_i+ a_x\tilde{u}_i \tilde{r}_i v_i\nonumber\\
    &\hphantom{{}=}+ Y(u_{d,i}, u_c)v_i^2 + a_y \tilde{u}_i v_i^2 - Y(u_{d,i}, u_c)v_c v_i - a_y \tilde{u}_i v_c v_i\nonumber\\
    &\hphantom{{}=} + X(u_{d,i}, u_c)\Big[|F(\cdot)| +|\kappa(\theta)|\left(|u_i| + |u_j|+ |v_j| \right)\Big] |v_i|.\label{eq:appendix:proofs:boundedness_x1_x2:v_dot_after_inserting_rd_v}
\end{align}
On the manifold where $(\tilde{\mathbf{X}}_1, \tilde{\mathbf{X}}_2) = \mathbf{0}$, the bound \eqref{eq:appendix:proofs:boundedness_x1_x2:v_dot_after_inserting_rd_v} simplifies to:
\begin{align}
    \dot{V} &\leq \left[|\kappa(\theta)|X_\text{max} + Y_\text{min}\right]v_i^2 \nonumber\\
    &\hphantom{{}=}+ X(u_{d,i}, u_c)\Big[F_1(\mathbf{0}, \mathbf{0}, \Delta, \theta, u_d, \dot{u}_d, v, v_c, u_v) +|\kappa(\theta)|\left(|u_i| + |u_j|+ |v_j| \right)\Big] |v_i|.\label{eq:appendix:proofs:boundedness_x1_x2:v_dot_on_manifold}
\end{align}
For sufficiently large $v_i$, we observe that the quadratic term is dominant. Consequently, boundedness of \eqref{eq:appendix:proofs:boundedness_x1_x2:v_dot_on_manifold} is guaranteed, since $\dot{V}$ is negative definite for sufficiently large $v_i$, that is:
\begin{equation}
    |\kappa(\theta)|X_\text{max} + Y_\text{min} < 0,
\end{equation}
which is satisfied whenever the maximum curvature satisfies \eqref{eq:nsb:analysis:barycenter:boundedness_X1_X2:max_curvature}. As $\dot{V}$ is negative definite for sufficiently large $v_i$, we can conclude that $V$ decreases for sufficiently large $v_i$. Furthermore, by extension, a decrease in $V$ implies a decrease in $v_i^2$ and again in $v_i$. Consequently, $v_i$ cannot increase above a certain threshold because this will make the quadratic term of \eqref{eq:appendix:proofs:boundedness_x1_x2:v_dot_on_manifold} dominant preventing further increase of $v_i$. Hence, $v_i$ is bounded near the manifold where  $(\tilde{\mathbf{X}}_1, \tilde{\mathbf{X}}_2) = \mathbf{0}$, concluding the proof of \cref{LEM:NSB:ANALYSIS:BARYCENTER:BOUNDEDNESS_X1_X2}.

\subsection{Proof of Lemma 3}
This proof follows along the lines of \cite[Lemma 3]{Belleter2018Proofs} but is extended to two vessels described by \eqref{eq:model:component_form} expressed in terms of absolute velocities, with the adaptive controllers \eqref{eq:nsb:analysis:heading_controller} - \eqref{eq:nsb:analysis:surge_controller}.

To prove boundedness of $v$ near the manifold $\tilde{\mathbf{X}}_2 = \mathbf{0}$, recall the sway dynamics \eqref{eq:nsb:analysis:barycenter:closed_loop:sway}:
\begin{align}
    \dot{v}_i &= X(u_{d,i} + \tilde{u}_i, u_c)r_{d_i} + X(u_{d,i} + \tilde{u}_i, u_c)\tilde{r}_i \nonumber\\
    &\hphantom{{}=}+ Y(u_{d,i} + \tilde{u}_i, u_c)v_i - Y(u_{d,i} + \tilde{u}_i, u_c)v_c.
\end{align}
We then consider the Lyapunov function candidate $V(v_i) = \frac{1}{2}v_i^2$, whose time derivative along the solutions of \eqref{eq:nsb:analysis:barycenter:closed_loop:sway} is:
\begin{align}
    \dot{V} &= v_i\dot{v}_i =  X(u_{d,i} + \tilde{u}_i, u_c)r_{d_i}v_i + X(u_{d,i} + \tilde{u}_i, u_c)\tilde{r}_i v_i \nonumber\\
    &\hphantom{{}=}+ Y(u_{d,i} + \tilde{u}_i, u_c)v_i^2 - Y(u_{d,i} + \tilde{u}_i, u_c)v_c v_i\\
    &\leq X(u_{d,i}, u_c)r_{d_i} v_i + a_x\tilde{u}_i r_{d,i} v_i + X(u_{d,i}, u_c)\tilde{r}_i v_i+ a_x\tilde{u}_i \tilde{r}_i v_i\nonumber\\
    &\hphantom{{}=}+ Y(u_{d,i}, u_c)v_i^2 + a_y \tilde{u}_i v_i^2 - Y(u_{d,i}, u_c)v_c v_i - a_y \tilde{u}_i v_c v_i.\label{eq:appendix:proofs:boundedness_x2:v_dot_before_inserting_rd_v}
\end{align}
Here, we have used the following properties of  $X(u, u_c)$ and $Y(u, u_c)$ from \eqref{eq:appendix:vessel_model_expressions:X} - \eqref{eq:appendix:vessel_model_expressions:Y}:
\begin{align}
    X(u, u_c) &= a_x u + b_x u_c + c_x \\
    Y(u, u_c) &= a_y u + b_y u_c + c_y.
\end{align}
Next, to find an upper bound of the term $r_{d,i} v_i$ in \eqref{eq:appendix:proofs:boundedness_x2:v_dot_before_inserting_rd_v}, we substitute the expression for $r_d$ from \eqref{eq:nsb:analysis:barycenter:r_d}, omitting subscripts for simplicity:
\begin{align}
    r_d v &= \kappa(\theta) \dot{\theta} v + \frac{\dot{u}_d }{u_d^2 + v^2} v^2 - \frac{u_d \dot{v}v}{u_d^2 + v^2}\nonumber\\
    &\hphantom{{}=}- \frac{v}{\Delta^2 + \left(y_{pb}^p\right)^2}\left[\Delta \dot{y}_{pb}^p - y_{pb}^p\left(\frac{\partial \Delta}{\partial x_{pb}^p}\dot{x}_{pb}^p + \frac{\partial \Delta}{\partial y_{pb}^p}\dot{y}_{pb}^p\right)\right] \\
    &= \kappa(\theta) v \left(\frac{1}{2}U_1\cos\left(\chi_1 - \gamma_p\right) + \frac{1}{2}U_2\cos\left(\chi_2 - \gamma_p\right) + \frac{k_\theta x_{pb}^p}{\sqrt{1 + \left(x_{pb}^p\right)^2}} \right) \nonumber\\
    &\hphantom{{}=} + \frac{\dot{u}_d }{u_d^2 + v^2} v^2 - \frac{u_dv}{u_d^2 + v^2}\Big(X(u, u_c) r + Y(u,u_c)v - Y(u,u_c)v_c\Big)\nonumber\\
    &\hphantom{{}=} -\frac{\Delta v}{\Delta^2 + \left(y_{pb}^p\right)^2}\left(- \frac{1}{2}\left(U_{d,1} + U_{d,2}\right)\frac{y_{pb}^p}{\sqrt{\Delta^2 + \left(y_{pb}^p\right)^2}} - \kappa(\theta)\dot{\theta}x_{pb}^p + G_1(\cdot)\right)\nonumber\\
    &\hphantom{{}=} +\frac{y_{pb}^p v}{\Delta^2 + \left(y_{pb}^p\right)^2}\bigggg[ \frac{\partial \Delta}{\partial x_{pb}^p}\left(- \frac{k_\theta x_{pb}^p}{\sqrt{1 + \left(x_{pb}^p\right)^2}} + \dot{\theta}\kappa(\theta)y_{pb}^p\right)\nonumber\\
    &\hphantom{{}=} + \frac{\partial \Delta}{\partial y_{pb}^p}\left(- \frac{1}{2}\left(U_{d,1} + U_{d,2}\right)\frac{y_{pb}^p}{\sqrt{\Delta^2 + \left(y_{pb}^p\right)^2}} - \kappa(\theta)\dot{\theta}x_{pb}^p + G_1(\cdot)\right)\bigggg].\label{eq:appendix:proofs:boundedness_x2:r_dv_before_extracting_H}
\end{align}
Now, introduce a term $H(\tilde{\mathbf{X}}_1, \tilde{\mathbf{X}}_2, \Delta, \theta, u_d, \dot{u}_d, v, v_c, u_c, r)$ to collect all terms that have less than quadratic growth in $v$ and/or vanish when $\tilde{\mathbf{X}}_2 = 0$.
\begin{align}
    r_d v &= \kappa(\theta) v \left[1 + \frac{x_{pb}^p}{\Delta^2 + \left(y_{pb}^p\right)^2}\right]\left(\frac{1}{2}U_1\cos\left(\chi_1 - \gamma_p\right) + \frac{1}{2}U_2\cos\left(\chi_2 - \gamma_p\right)\right)\nonumber\\
    &\hphantom{{}=} -\frac{\Delta v}{\Delta^2 + \left(y_{pb}^p\right)^2}\left(- \frac{1}{2}\left(U_{d,1} + U_{d,2}\right)\frac{y_{pb}^p}{\sqrt{\Delta^2 + \left(y_{pb}^p\right)^2}} + G_1(\cdot)\right)\nonumber\\
    &\hphantom{{}=} -\frac{\Delta y_{pb}^p}{\Delta^2 + \left(y_{pb}^p\right)^2} \frac{\partial \Delta}{\partial y_{pb}^p}\left(- \frac{1}{2}\left(U_{d,1} + U_{d,2}\right)\frac{y_{pb}^p}{\sqrt{\Delta^2 + \left(y_{pb}^p\right)^2}} + G_1(\cdot)\right)\nonumber\\
    &\hphantom{{}=} -\frac{u_d}{u_d^2 + v^2} Y(u, u_c)v^2 + H(\cdot),
\end{align}
where
\begin{align}
    H(\cdot) &= v\bigggg[ \kappa(\theta)\frac{k_\theta x_{pb}^p}{\sqrt{1 + \left(x_{pb}^p\right)^2}} + \frac{\dot{u}_d v}{u_d^2 + v^2} - \frac{u_dv}{u_d^2 + v^2}\Big(X(u, u_c) r - Y(u,u_c)v_c\Big)\nonumber\\
    &\hphantom{{}=} + \frac{\Delta}{\Delta^2 + \left(y_{pb}^p\right)^2}\kappa(\theta) \frac{k_\theta \left(x_{pb}^p\right)^2}{\sqrt{1 + \left(x_{pb}^p\right)^2}} - \frac{y_{pb}^p}{\Delta^2 + \left(y_{pb}^p\right)^2}\frac{\partial \Delta}{\partial x_{pb}^p} \frac{k_\theta x_{pb}^p}{\sqrt{1 + \left(x_{pb}^p\right)^2}}\bigggg].
\end{align}
Similarly as in the proof of \cref{LEM:NSB:ANALYSIS:BARYCENTER:BOUNDEDNESS_X1_X2}, all terms with partial derivatives of $\Delta$ and $\dot{\theta}$ are cancelled due to skew-symmetry. Consequently, the term $r_d v$ may be upper bounded as:
\begin{align}
    r_{d,i} v_i &\leq |v_i| \left|\left[ 1 + \frac{\Delta x_{pb}^p}{\Delta^2 + \left(y_{pb}^p\right)^2}\right]\right| |\kappa(\theta)| \frac{1}{2}\left(|U_i| + |U_j|\right) + |v_i| \left|\frac{1}{\Delta}\right|\frac{1}{2}\left(|U_{d,i}| + |U_{d,j}| + |G_1(\cdot)|\right)\nonumber\\
    &\hphantom{{}=} + |v_i|\left|\frac{y_{pb}^p}{\Delta^2 + \left(y_{pb}^p\right)^2}\right|\frac{1}{2}\left(|U_{d,i}| + |U_{d,j}| + |G_1(\cdot)|\right) -\frac{u_{d,i}}{u_{d,i}^2 + v_i^2} Y(u_i, u_c)v_i^2 + |H(\cdot).\label{eq:appendix:proofs:boundedness_x2:first_bound_r_dv}
\end{align}
To further restrict the upper bound on $r_d v$, we will utilize the following inequalities:
\begin{align}
    \left|\frac{y_{pb}^p}{\Delta^2 + \left(y_{pb}^p\right)^2}\right| &\leq \left|\frac{1}{\Delta}\right| \label{eq:appendix:proofs:boundedness_x2:inequality_1}\\
    |U_{d,i}| &\leq 4 (|u_i| + |v_i| + |\tilde{u}_i|.\label{eq:appendix:proofs:boundedness_x2:inequality_2}
\end{align}
Substituting \eqref{eq:appendix:proofs:boundedness_x2:inequality_1} - \eqref{eq:appendix:proofs:boundedness_x2:inequality_2} into \eqref{eq:appendix:proofs:boundedness_x2:first_bound_r_dv} we obtain:
\begin{align}
    r_{d,i} v_i &\leq \frac{1}{2} v_i^2 \left|\left[ 1 + \frac{\Delta x_{pb}^p}{\Delta^2 + \left(y_{pb}^p\right)^2}\right]\right| |\kappa(\theta)| + \frac{1}{2} |v_i| \left|\left[ 1 + \frac{\Delta x_{pb}^p}{\Delta^2 + \left(y_{pb}^p\right)^2}\right]\right| |\kappa(\theta)| |u_i|\nonumber\\
    &\hphantom{{}=} + \frac{1}{2}|v_i| \left|\left[ 1 + \frac{\Delta x_{pb}^p}{\Delta^2 + \left(y_{pb}^p\right)^2}\right]\right| |\kappa(\theta)| |U_j| -\frac{u_{d,i}}{u_{d,i}^2 + v_i^2} Y(u_i, u_c)v_i^2 + |H(\cdot)| \nonumber\\
    &\hphantom{{}=} + |v_i| \left|\frac{2}{\Delta}\right|\left(2\left(|u_i| + |v_i|\right) + \frac{1}{2}|\tilde{u}_i| + 2\left(|u_j| + |v_j|\right) + \frac{1}{2}|\tilde{u}_j| + |G_1(\cdot)| \right).\label{eq:appendix:proofs:boundedness_x2:bound_r_dv_before_Phi}
\end{align}
Then, we substitute \eqref{eq:appendix:proofs:boundedness_x1_x2:1_plus_delta_x_less_than_two} into \eqref{eq:appendix:proofs:boundedness_x2:bound_r_dv_before_Phi} and introduce the function $\Phi(\cdot)$ to collect the remaining terms that have less than quadratic growth in $v_i$ and/or vanish when $\tilde{\mathbf{X}}_2 = 0$:
\begin{align}
    r_{d,i} v_i &\leq v_i^2 \left[\frac{1}{2} |\kappa(\theta)| \left|\left[ 1 + \frac{\Delta x_{pb}^p}{\Delta^2 + \left(y_{pb}^p\right)^2}\right]\right| + \frac{4}{\Delta}\right] - \frac{u_{d,i}}{u_{d,i}^2 + v_i^2} Y(u_i, u_c)v_i^2 + \Phi(\cdot)\\
    &\leq v_i^2\left[|\kappa(\theta)| + \frac{4}{\Delta}\right] - \frac{u_{d,i}}{u_{d,i}^2 + v_i^2} Y(u_i, u_c)v_i^2 + \Phi(\cdot).\label{eq:appendix:proofs:boundedness_x2:bound_r_dv_with_Phi}
\end{align}
From the definitions of $\Phi(\cdot)$ and $H(\cdot)$ we can conclude the existence of some positive bounded functions \\$F_{0,2}(\tilde{\mathbf{X}}_1, \tilde{\mathbf{X}}_2, \Delta, \theta, u_d, \dot{u}_d, v, v_c, u_c, r)$, such that:
\begin{gather}
    \Phi(\cdot) \leq F_2(\cdot) v_i ^2 + F_2(\cdot) |v_i| + F_0(\cdot)\\
    F_2(\tilde{\mathbf{X}}_1, \mathbf{0}, \Delta, \theta, u_d, \dot{u}_d, v, v_c, u_c, r) = 0.
\end{gather}
Having established an upper bound of the term $r_d v$, we substitute \eqref{eq:appendix:proofs:boundedness_x2:bound_r_dv_with_Phi} into \eqref{eq:appendix:proofs:boundedness_x2:v_dot_before_inserting_rd_v} to obtain the following bound for $\dot{V}$:
\begin{align}
    \dot{V} &\leq X(u_{d,i}, u_c)\left(\left[|\kappa(\theta)| + \frac{4}{\Delta}\right]v_i^2 + \Phi(\cdot)\right) + a_x\tilde{u}_i r_{d,i} v_i + X(u_{d,i}, u_c)\tilde{r}_i v_i+ a_x\tilde{u}_i \tilde{r}_i v_i\nonumber\\
    &\hphantom{{}=}+ Y(u_{d,i}, u_c)v_i^2 + a_y \tilde{u}_i v_i^2 - Y(u_{d,i}, u_c)v_c v_i - a_y \tilde{u}_i v_c v_i - \frac{u_{d,i}}{u_{d,i}^2 + v_i^2} Y(u_i, u_c)v_i^2\\
    &\leq \left(X(u_{d,i}, u_c)\left[|\kappa(\theta)| + \frac{4}{\Delta}\right]-Y(u_{d,i}, u_c)\right)v_i^2 + a_x\tilde{u}_i r_{d,i} v_i + X(u_{d,i}, u_c)\tilde{r}_i v_i+ a_x\tilde{u}_i \tilde{r}_i v_i\nonumber\\
    &\hphantom{{}=} + a_y \tilde{u}_i v_i^2 - Y(u_{d,i}, u_c)v_c v_i - a_y \tilde{u}_i v_c v_i + X(u_{d,i}, u_c) \Phi(\cdot). \label{eq:appendix:proofs:boundedness_x2:v_dot_after_inserting_rd_v}
\end{align}
On the manifold where $\tilde{\mathbf{X}}_2 = 0$, the bound \eqref{eq:appendix:proofs:boundedness_x2:v_dot_after_inserting_rd_v} simplifies to:
\begin{align}
    \dot{V} &\leq \left(X_\text{max}\left[\kappa_\text{max} + \frac{4}{\Delta}\right]-Y_\text{min}\right)v_i^2\nonumber\\
    &\hphantom{{}=} +X(u_{d,i}, u_c)\Big(F_1(\tilde{\mathbf{X}}_1, \mathbf{0}, \Delta, \theta, u_d, \dot{u}_d, v, v_c, u_c, r)|v_i| + F_0(\tilde{\mathbf{X}}_1, \mathbf{0}, \Delta, \theta, u_d, \dot{u}_d, v, v_c, u_c, r)\Big). \label{eq:appendix:proofs:boundedness_x2:v_dot_on_manifold}
\end{align}
For sufficiently large $v_i$, we observe that the quadratic term is dominant. Consequently, boundedness of \eqref{eq:appendix:proofs:boundedness_x2:v_dot_on_manifold} is guaranteed, since $\dot{V}$ is negative definite for sufficiently large $v_i$, whenever the following condition holds:
\begin{equation}
    X_\text{max}\left[\kappa_\text{max} + \frac{4}{\Delta}\right]-Y_\text{min} < 0.
\end{equation}
Using the lookahead distance definition in \eqref{EQ:NSB:TASKS:BARYCENTER:GUIDANCE:LOOKAHEAD}, we conclude that this condition is fulfilled whenever the conditions of \cref{LEM:NSB:ANALYSIS:BARYCENTER:BOUNDEDNESS_X2} is fulfilled:
\begin{equation}
    \mu > \frac{4 X_\text{max}}{Y_\text{min} -\kappa_\text{max} X_\text{max}}.
\end{equation}
Note how this is well defined as the denominator is nonzero and positive whenever the condition from \cref{LEM:NSB:ANALYSIS:BARYCENTER:BOUNDEDNESS_X1_X2} is satisfied. As $\dot{V}$ is negative definite for sufficiently large mangnitudes of $v_i$ near the manifold $\tilde{\mathbf{X}}_2 = 0$, the Lyapunov function candidate $V(v_i) = \frac{1}{2}v_i^2$ must decrease for sufficiently large $v_i$, and by extension, the magnitude of $v_i$ must decrease for sufficiently large $v_i$. Hence, $v_i$ is bounded near the manifold where $\tilde{\mathbf{X}}_2 = 0$ if the constant part of the lookahead distance $\mu$ is chosen accordingly to the condition in \cref{LEM:NSB:ANALYSIS:BARYCENTER:BOUNDEDNESS_X2}.

\subsection{Proof that the lookahead distance can be chosen independently of $y_{pb}^p$}
\label{subsec:appendix:proofs:boundedness_x1_x2:lookahead}

In this section we will prove that the lookahead distance $\Delta$ in \eqref{EQ:NSB:TASKS:BARYCENTER:GUIDANCE:LOOKAHEAD} is required to be dependent on $x_{pb}^p$, but that it can be chosen independently of $y_{pb}^p$. This is in contrast to \cite{Belleter2019} where the lookahead distance was required to be a function of both $x_{b/p}$ and  $y_{b/p}$.

First, we will investigate the consequences of choosing the lookahead distance in \eqref{EQ:NSB:TASKS:BARYCENTER:GUIDANCE:LOOKAHEAD} independently of both $x_{pb}^p$ and $y_{pb}^p$. In this case, the expression for $r_d v$ in \eqref{eq:appendix:proofs:boundedness_x1_x2:r_dv_before_extracting_F} reduce to
\begin{align}
    r_d v &= \kappa(\theta) v \left(\frac{1}{2}U_1\cos\left(\chi_1 - \gamma_p\right) + \frac{1}{2}U_2\cos\left(\chi_2 - \gamma_p\right) + \frac{k_\theta x_{pb}^p}{\sqrt{1 + \left(x_{pb}^p\right)^2}} \right) \nonumber\\
    &\hphantom{{}=} + \frac{\dot{u}_d }{u_d^2 + v^2} v^2 - \frac{u_dv}{u_d^2 + v^2}\Big(X(u, u_c) r + Y(u,u_c)v - Y(u,u_c)v_c\Big)\nonumber\\
    &\hphantom{{}=} -\frac{\Delta v}{\Delta^2 + \left(y_{pb}^p\right)^2}\left(- \frac{1}{2}\left(U_{d,1} + U_{d,2}\right)\frac{y_{pb}^p}{\sqrt{\Delta^2 + \left(y_{pb}^p\right)^2}} - \kappa(\theta)\dot{\theta}x_{pb}^p + G_1(\cdot)\right),
\end{align}
as the partial derivatives of $\Delta$ with respect to $x_{pb}^p$ and $y_{pb}^p$ will be zero. Now, we want to take a closer look at the term: 
\begin{equation}
    v \frac{\Delta\kappa(\theta)\dot{\theta}x_{pb}^p}{\Delta^2 + \left(y_{pb}^p\right)^2}.
\end{equation}
Focusing on the $i$\textsuperscript{th} vessel. Inserting the expression for $\dot{\theta}$, and isolating the part independent of $x_{pb}^p$ and depending on the part $U_{d,i}$ we get:
\begin{equation}
    v_i \frac{\Delta\kappa(\theta)x_{pb}^p}{\Delta^2 + \left(y_{pb}^p\right)^2} \frac{1}{2}U_i\cos\left(\chi_i - \gamma_p\right). \label{eq::appendix:proofs:boundedness_x1_x2:lookahead:independent_of_x_and_y_before_inserting_for_U_i}
\end{equation}
From the definition $U_i = \sqrt{u_i^2 + v_i^2}$, the growth of $U_i$ is proportional with $v_i$. Thus, the growth of \eqref{eq::appendix:proofs:boundedness_x1_x2:lookahead:independent_of_x_and_y_before_inserting_for_U_i} can be represented by:
\begin{equation}
    v_i^2 \frac{\Delta\kappa(\theta)x_{pb}^p}{\Delta^2 + \left(y_{pb}^p\right)^2}. \label{eq::appendix:proofs:boundedness_x1_x2:lookahead:independent_of_x_and_y_after_inserting_for_U_i}
\end{equation}
If $\Delta$ is chosen independently of $x_{pb}^p$ it is clear how this term will go to infinity for large values of $x_{pb}^p$. However, if $\Delta$ is chosen to grow at least linearly with $x_{pb}^p$, then \eqref{eq::appendix:proofs:boundedness_x1_x2:lookahead:independent_of_x_and_y_after_inserting_for_U_i} will not diverge as
\begin{equation}
    \lim_{x_{pb}^p \to\infty} \frac{\kappa(\theta)\left(x_{pb}^p\right)^2}{\left(x_{pb}^p\right)^2 + \left(y_{pb}^p\right)^2} = 1,
\end{equation}
or converge to zero if $\Delta$ grows more than linear in $x_{pb}^p$. Thus, if $\Delta$ is chosen independently of $x_{pb}^p$, $r_d v$ could grow unbounded with $x_{pb}^p$ making it impossible to show boundedness of the sway dynamics. 

Next, we will show that the lookahead distance can be chosen independently of $y_{pb}^p$, which is in contrast to \cite{Belleter2018Proofs} where it was shown that the lookahead distance also needed to be chosen dependent on $y_{b/p}$. Lets consider the case where the lookahead distance in \eqref{EQ:NSB:TASKS:BARYCENTER:GUIDANCE:LOOKAHEAD} depend on $x_{pb}^p$, but is independent of $y_{pb}^p$. Then, the expression for $r_d v$ in \eqref{eq:appendix:proofs:boundedness_x1_x2:r_dv_before_extracting_F} reduces to:
\begin{align}
    r_d v  &= \kappa(\theta) v \left(\frac{1}{2}U_1\cos\left(\chi_1 - \gamma_p\right) + \frac{1}{2}U_2\cos\left(\chi_2 - \gamma_p\right) + \frac{k_\theta x_{pb}^p}{\sqrt{1 + \left(x_{pb}^p\right)^2}} \right) \nonumber\\
    &\hphantom{{}=} + \frac{\dot{u}_d }{u_d^2 + v^2} v^2 - \frac{u_dv}{u_d^2 + v^2}\Big(X(u, u_c) r + Y(u,u_c)v - Y(u,u_c)v_c\Big)\nonumber\\
    &\hphantom{{}=} -\frac{\Delta v}{\Delta^2 + \left(y_{pb}^p\right)^2}\left(- \frac{1}{2}\left(U_{d,1} + U_{d,2}\right)\frac{y_{pb}^p}{\sqrt{\Delta^2 + \left(y_{pb}^p\right)^2}} - \kappa(\theta)\dot{\theta}x_{pb}^p + G_1(\cdot)\right)\nonumber\\
    &\hphantom{{}=} +\frac{y_{pb}^p v}{\Delta^2 + \left(y_{pb}^p\right)^2}\bigggg[ \frac{\partial \Delta}{\partial x_{pb}^p}\left(- \frac{k_\theta x_{pb}^p}{\sqrt{1 + \left(x_{pb}^p\right)^2}} + \dot{\theta}\kappa(\theta)y_{pb}^p\right)\bigggg].
\end{align}
Using the same approach as earlier, we isolate the term
\begin{equation}
    v \frac{\partial \Delta}{\partial x_{pb}^p}\kappa(\theta)\dot{\theta} \frac{\left(y_{pb}^p\right)^2}{\Delta^2 + \left(y_{pb}^p\right)^2},
\end{equation}
which with the same reasoning as before reduces to:
\begin{equation}
    v_i^2 \frac{\partial \Delta}{\partial x_{pb}^p}\kappa(\theta) \frac{\left(y_{pb}^p\right)^2}{\Delta^2 + \left(y_{pb}^p\right)^2}.\label{eq::appendix:proofs:boundedness_x1_x2:lookahead:independent_of_x}
\end{equation}
It is clear how this term is bounded for all values of $y_{pb}^p$ even when $\Delta$ depends only on $x_{pb}^p$. When comparing this to the term from \cite[Eq. (83)]{Belleter2018Proofs}
\begin{equation}
    \frac{v_r^2}{C_r} \frac{\partial \Delta}{\partial x_{b/p}}\kappa(\theta)\frac{\Delta y_{b/p}(y_{b/p} + g)}{\left(\Delta^2 + \left(y_{b/p} + g\right)^2\right)^{3/2}}, \label{eq:appendix:proofs:boundedness_x1_x2:lookahead:belleter2019Proofs_eq_83}
\end{equation}
it is clear how \eqref{eq:appendix:proofs:boundedness_x1_x2:lookahead:belleter2019Proofs_eq_83} can grow unbounded in $y_{b/p}$ near the manifold where $g = -(y_{b/p} + 1)$ as the term, near the manifold, reduces to
\begin{equation}
    \frac{v_r^2}{C_r} \frac{\partial \Delta}{\partial x_{b/p}}\kappa(\theta)\frac{\Delta y_{b/p}}{\left(\Delta^2 + 1\right)^{3/2}},
\end{equation}
which can only be bounded by choosing $\Delta$ dependent on $y_{b/p}$. 
However, from our choice of expressing the LOS guidance law in terms of absolute velocities, combined with a set of adaptive autopilots for ocean current compensation, the ocean current dependent term $g$ is not present in \eqref{eq:nsb:tasks:barycenter:guidance:guidance_law_heading}. Consequently, there are no manifold where \eqref{eq::appendix:proofs:boundedness_x1_x2:lookahead:independent_of_x} can grow unbounded in $y_{pb}^p$, implying the lookahead distance \eqref{EQ:NSB:TASKS:BARYCENTER:GUIDANCE:LOOKAHEAD} can be chosen independently of $y_{pb}^p$. 
\fi

\end{document}